\documentclass[conference,compsoc]{IEEEtran}
\usepackage{prologue, graphicx}
\usepackage{mdwlist}
\usepackage{epstopdf}
\usepackage{bm}
\usepackage{makecell}
\usepackage{color}
\usepackage{multirow}
\usepackage{enumitem}
\usepackage{bbm}
\usepackage{array}
\usepackage{soul}

\ifCLASSOPTIONcompsoc
  \usepackage[nocompress]{cite}
\else
  \usepackage{cite}
\fi
\ifCLASSINFOpdf
\else
\fi
\begin{document}
\title{From Randomized Response to Randomized Index: Answering Subset Counting Queries with Local Differential Privacy}

%\author{
%
%\IEEEauthorblockN{Qingqing Ye}
%\IEEEauthorblockA{Department of Electrical and \\ Electronic Engineering\\
%The Hong Kong Polytechnic University\\
%qqing.ye@polyu.edu.hk}
%
%\and
%
%\IEEEauthorblockN{Liantong Yu}
%\IEEEauthorblockA{Department of Electrical and \\ Electronic Engineering\\
%	The Hong Kong Polytechnic University\\
%	liantong2001.yu@connect.polyu.hk}
%
%\and
%
%\IEEEauthorblockN{Kai Huang}
%\IEEEauthorblockA{Faculty of Information Technology \\
%	Macau University of Science \\ and Technology\\
%	kylehuangk@gmail.com}
%
%
%\and
%\IEEEauthorblockN{Xiaokui Xiao}
%\IEEEauthorblockA{School of Computing \\
%	National University of Singapore \\
%	xkxiao@nus.edu.sg}
%
%\and 
%
%\IEEEauthorblockN{Weiran Liu}
%\IEEEauthorblockA{Alibaba Group \\
%		weiran.lwr@alibaba-inc.com}
%
%\and
%
%\IEEEauthorblockN{Haibo Hu}
%\IEEEauthorblockA{Department of Electrical and \\ Electronic Engineering\\
%	The Hong Kong Polytechnic University\\
%	haibo.hu@polyu.edu.hk}
%}


% compsoc conferences compared to traditional conferences), use this

\IEEEoverridecommandlockouts

\author{
	\IEEEauthorblockN{
Qingqing Ye\IEEEauthorrefmark{1},
Liantong Yu\IEEEauthorrefmark{1},
Kai Huang\IEEEauthorrefmark{2},
Xiaokui Xiao\IEEEauthorrefmark{3}, 
Weiran Liu\IEEEauthorrefmark{4},
%Haibo Hu\IEEEauthorrefmark{1}
Haibo Hu\IEEEauthorrefmark{1}
}
\IEEEauthorblockA{\IEEEauthorrefmark{1}Department of Electrical and Electronic Engineering, 
The Hong Kong Polytechnic University\\ Email: qqing.ye@polyu.edu.hk, liantong2001.yu@connect.polyu.hk, haibo.hu@polyu.edu.hk}
\IEEEauthorblockA{\IEEEauthorrefmark{2}Faculty of Information Technology, Macau University of Science and Technology\\
Email: kylehuangk@gmail.com}
\IEEEauthorblockA{\IEEEauthorrefmark{3}School of Computing,	National University of Singapore \\
	Email: xkxiao@nus.edu.sg}
\IEEEauthorblockA{\IEEEauthorrefmark{4}Alibaba Group \\
	Email: weiran.lwr@alibaba-inc.com}
}

%\author{\IEEEauthorblockN{Qingqing Ye\IEEEauthorrefmark{1},
%		Liantong Yu\IEEEauthorrefmark{1},
%		Kai Huang\IEEEauthorrefmark{2},
%		Xiaokui Xiao\IEEEauthorrefmark{3}, 
%		Ruiran Liu\IEEEauthorrefmark{4} and
%		Haibo Hu\IEEEauthorrefmark{1}}
%	\IEEEauthorblockA{\IEEEauthorrefmark{1}Department of Electrical and Electronic Engineering, 
%		The Hong Kong Polytechnic University}
%	\IEEEauthorblockA{\IEEEauthorrefmark{2}Faculty of Information Technology, Macau University of Science and Technology}
%	\IEEEauthorblockA{\IEEEauthorrefmark{3}School of Computing,	National University of Singapore}
%	\IEEEauthorblockA{\IEEEauthorrefmark{4}Alibaba Group}
%	qqing.ye@polyu.edu.hk, liantong2001.yu@connect.polyu.hk, kylehuangk@gmail.com, \\ xkxiao@nus.edu.sg, weiran.lwr@alibaba-inc.com, haibo.hu@polyu.edu.hk
%}


\maketitle

\begin{abstract}
	Local Differential Privacy (LDP) is the predominant privacy model for safeguarding individual data privacy. Existing perturbation mechanisms typically require perturbing the original values to ensure acceptable privacy, which inevitably results in value distortion and utility deterioration. In this work, we propose an alternative approach --- instead of perturbing values, we apply randomization to indexes of values while ensuring rigorous LDP guarantees.	
	
	Inspired by the deniability of randomized indexes, we present CRIAD for answering subset counting queries on set-value data. By integrating a multi-dummy, multi-sample, and multi-group strategy, CRIAD serves as a fully scalable solution that offers flexibility across various privacy requirements and domain sizes, and achieves more accurate query results than any existing methods. Through comprehensive theoretical analysis and extensive experimental evaluations, we validate the effectiveness of CRIAD and demonstrate its superiority over traditional value-perturbation mechanisms.
\end{abstract}


% \fi
\IEEEpeerreviewmaketitle

\section{Introduction}
\label{sec:introduction}
As big data analytics become more prevalent, service providers are increasingly eager to collect extensive usage data to improve their services. However, much of this data, particularly when gathered from individuals, includes sensitive information such as biometrics, personal identification, health data, financial transactions, and location trajectories. Directly collecting such data for model training or statistical analysis raises significant privacy concerns.

To address these privacy challenges, Local Differential Privacy  (LDP)~\cite{duchi2013local,kasiviswanathan2011can} has emerged as the predominant privacy model in many large-scale distributed systems, used by tech giants such as Apple~\cite{thakurta2017learning}, Google~\cite{erlingsson2014rappor} and Microsoft~\cite{ding2017collecting}, to protect end users' data. Despite its benefits, LDP has faced criticism for its low utility, as the values collected must undergo significant or even unbounded perturbation by noise injection~\cite{dwork2006calibrating,dwork2014algorithmic,wang2019collecting} or value randomization~\cite{warner1965randomized,kairouz2014extremal,wang2017locally} to ensure acceptable privacy.

In the literature, there are many existing works on LDP for set-value data~\cite{qin2016heavy, wang2018locally, wang2018privset, wang2019locally, ye2019privkv, du2024top}. Specifically, each user possesses a set of private items, and the data collector aims to estimate the frequency of a specific item or identify the most frequent items among users. However, in real-world applications, items are often organized into categories and there is a need to estimate statistics for the set of items belonging to a specific category. The following are two examples.

\begin{itemize}
	\item 	
	Amazon sells millions of books, each belonging to a specific category, e.g., fiction, poetry. To predict sales and manage inventory effectively, Amazon needs to obtain the sales volume for a specific category among thousands of titles.
	
	\item
 	To optimize advertising strategies, IMDb needs to determine which movie genre attracts the largest audience. This involves analyzing the number of users interested in a specific genre (e.g., thriller) that encompasses a set of movies.
\end{itemize}

In this study, we formulate such problem as a {\bf subset counting query}, which returns the total count of items within a subset of the item domain. To answer this query in the context of LDP, there are two solutions adapted from existing works. First, each user counts her items belonging to that subset and then employs a numerical-value perturbation mechanism (e.g., Laplace Mechanism~\cite{dwork2006calibrating} or Piecewise Mechanism~\cite{wang2019collecting}) to inject random noise to the count and then reports the sanitized count. Alternatively, each user employs a categorical-value perturbation mechanism (e.g., Randomized Response~\cite{warner1965randomized, kairouz2014extremal}) to perturb her item set and then reports the perturbed items.

However, both two solutions introduce perturbation errors into the original value. In mission-critical applications, particularly in medical and financial applications, value perturbation does not apply as distorted values become useless or even detrimental. In fact, significant distortion of values can overshadow the original value and potentially make it unbounded. In this work, we propose an alternative approach: instead of perturbing values, we apply randomization to indexes of values, while ensuring a rigorous LDP guarantee. The following example illustrates how randomized index ensures plausible deniability, while the original values remain intact.

{\bf Example.} {\it In a ballot with $10$ candidates, each voter independently votes yes/no to each candidate. An interviewer wants to estimate the average number of ``yes" votes of all voters. To ensure privacy of these votes, each voter randomly samples a candidate index from $\{1, 2, ..., 10\}$, and then faithfully reports her vote (i.e., yes/no) for the selected candidate, but without indicating the index she has sampled. }

Inspired by the deniability of the above randomized index, we propose CRI (short for \underline{C}ounting via \underline{R}andomized \underline{I}ndex) protocol for answering subset counting queries, while satisfying $\epsilon$-LDP. Although perturbation noise is not necessary in most cases, to ensure an unbiased estimation, CRI still suffers from utility loss by re-introducing certain perturbation noise in few extreme cases. To address this issue, we develop a dummy-assisted solution CRIAD (short for \underline{CRI} with \underline{A}ugmented \underline{D}ummies) to eliminate perturbation noise, which thus achieves significantly higher accuracy. On the other hand, we further enhance the scalability of CRIAD by developing a multi-dummy, multi-sample and multi-group strategy that can support a wide range of privacy requirements (specified by $\epsilon$) and domain sizes. Through theoretical analysis and empirical studies, we show the effectiveness of CRIAD. 
To summarize, our contributions are three-folded.
\begin{itemize}
	\item
	We formulate the problem of answering subset counting queries under LDP, and design randomized-index-based solutions that can ensure rigorous LDP guarantees.
	To the best of our knowledge, this is the first LDP mechanism based on the deniability of randomized indexes.
	
	\item
	We develop a scalable solution, CRIAD, which accommodates flexible privacy requirements and domain sizes. By leveraging augmented dummy items, CRIAD eliminates perturbation noise injected into the original value while satisfying $\epsilon$-LDP.    	
	
	\item
	Through comprehensive theoretical and experimental analysis, we demonstrate that CRIAD achieves significantly higher accuracy than existing value-perturbation LDP mechanisms.
\end{itemize}

The rest of the paper is organized as follows. Section~\ref{sec:background} introduces preliminaries, problem definition, and existing solutions. Section~\ref{sec:baseline} presents our baseline solution CRI via randomized index deniability. Section~\ref{sec:rid} proposes the optimized solution CRIAD. Section~\ref{sec:experiment} presents experimental evaluations. Finally, we review existing work in Section~\ref{sec:related_work} and conclude this paper in Section~\ref{sec:conclusion}.

%\section{Background}
%\label{sec:background}
%\input{background}
%
%\section{Problem Definition and Existing Solutions}
%\label{sec:problem}
%\input{problem}

\section{Preliminaries and Problem Definition}
\label{sec:background}
In this section, we introduce preliminaries on LDP, formulate the problem of answering a subset counting query under LDP, and then present two naive solutions that are directly adapted from existing works.

\subsection{Local Differential Privacy}
Differential privacy (DP)~\cite{dwork2006differential,dwork2006calibrating} works in both centralized and local settings. Centralized DP~\cite{li2016differential} requires the data curator to be fully trusted to collect all data, while local DP does not rely on this assumption. In the local setting~\cite{duchi2013local,kasiviswanathan2011can}, each user locally perturbs her data before reporting them to an untrusted data collector, which makes it more secure and practical in real-world applications. The formal definition is as follows.
\begin{definition}
	\label{def:ldp}
	({\bf Local Differential Privacy, LDP}) A randomized algorithm $\mathcal{A}$ satisfies $\epsilon$-LDP if for any two input records $w$ and $w'$, and any set $W$ of possible outputs of $\mathcal{A}$, the following inequality holds.
	\begin{align}
		\label{eq:ldp}
		\frac{\mathrm{Pr}[\mathcal{A}(w)\in W]}{\mathrm{Pr}[\mathcal{A}(w') \in W]}\le e^\epsilon 
	%	\mathrm{Pr}[\mathcal{A}(w)\in W] \le e^\epsilon \times \mathrm{Pr}[\mathcal{A}(w') \in W]
	\end{align}
\end{definition}

In the above definition, $\epsilon$ is called the privacy budget, which controls the deniability of a randomized algorithm taking $w$ or $w'$ as its input. As with centralized DP, LDP also has the property of sequential composition~\cite{mcsherry2009privacy} as below, which guarantees the overall privacy for a sequence of randomized algorithms.
\begin{theorem}
	\label{theorem:composition}
	({\bf Sequential Composition}) Given $t$ randomized algorithms $\mathcal{A}_i(1\leq i\leq t)$, each providing $\epsilon_i$-local differential privacy, then the sequence of algorithms $\mathcal{A}_i(1\leq i\leq t)$ collectively provides $(\Sigma \epsilon_i)$-local differential privacy.
\end{theorem}

\subsection{Problem Definition}
\label{sec:definition}
We assume there are $n$ users $\mathcal{U}=\{u_1, u_2, ..., u_n\}$, and each user $u_i$ possesses a set of private items $R_i \subseteq D$, where $D=\{r_1, r_2, ..., r_{|D|}\}$ is the item domain. The domain has some categories $\{c_1, c_2, ...\}$, and each item belongs to one or more categories. In other words, items belonging to a category is a subset of the domain. For ease of presentation, we use $d=|c_i|$ to denote the sub-domain size of category $c_i$, i.e., the number of items in the domain $D$ belonging to category $c_i$. 
Now we formally define the subset counting query as follows.

\begin{definition}
	({\bf Subset Counting Query}) Denoted by $Q(c)$, a subset counting query takes as input a category $c$, and returns the count of items belonging to $c$ among the user population. Formally,
	\begin{align}
		\label{eq:subset_counting_query}
		Q(c) = \sum_{i=1}^{n}\sum_{r\in R_i} \mathbbm{1}_c(r),
	\end{align}
	where $\mathbbm{1}_c(r)$ is an indicator function, which returns $1$ if an item $r$ belongs to the category $c$, and returns $0$ otherwise. 
\end{definition}

%We summarize the notations in Table~\ref{table:notation}.

\begin{table}
	\caption{Notations}
	\label{table:notation}
	\centering
	\begin{tabular}{|c|c|c|}
		\hline
		\bf Symbol          & \bf Description   \\ \hline
		$\epsilon$      & the privacy budget \\ \hline
 		$n$                 & the number of users        \\ \hline
		$u_i$    			& the $i$-th user in user population, $i\in \{1,..., n\}$ \\ \hline		
		$D$                 & item domain \\ \hline
		$R_i$    			& a set of items possessed by user $u_i$, $R_i \subseteq D$ \\ \hline
		$c_i$               & the $i$-th category of items in the domain, $i\in \{1, 2, ...\}$ \\ \hline		
		$d$             & the size of category, $d=|c_i|$ \\ \hline
		$m$	                & the number of dummy items in CRIAD \\ \hline
		$g$	                & the number of groups in CRIAD \\ \hline
		$s$	                & the number of samples in CRIAD \\ \hline
	\end{tabular}
\end{table}

\subsection{Solutions from Existing Work}
\label{sec:numerical_ldp}
To answer a subset counting query $Q(c)$ of a category $c$, there are two naive solutions adapted from existing works.

{\bf Solution 1: Numerical Value Perturbation (NVP)}.
Each user $u_i$ locally counts her items that belong to category $c$. Then she perturbs it and reports a sanitized count $t_i^*$ by a numerical-value perturbation mechanism $\mathcal{A}(\cdot)$, e.g., Laplace Mechanism (LM)~\cite{dwork2006calibrating}, Piecewise Mechanism (PM)~\cite{wang2019collecting} or Square Wave mechanism (SW)~\cite{li2020estimating}. Formally,
\begin{align*}
	t_i^* = \mathcal{A}\left(\sum\nolimits_{r\in R_i} \mathbbm{1}_c(r)\right),
\end{align*}
Then the subset counting query result is the summation of noisy reports from all users, i.e., $\sum_{i=1}^{n}t_i^*$.

%{\bf Solution 2: Sampling RR}.
%The second solution is to let each user $u_i$ first encodes her record set concerning the category $C_j^k$ into a binary vector $V_i=\{V_i^1, V_i^2, ..., V_i^d\}$, where $d=|C_j^k|$ and for each $l\in \{1, 2, ..., d\}$, $V_i^l=1$ if she has the $l$-th value of $C_j^k$, and $V_i^l=0$ otherwise. Then, she randomly samples a bit (i.e., $0$ or $1$) and perturbs it by Randomized Response (RR)~\cite{warner1965randomized}. Specifically, given privacy budget $\epsilon$, the user first randomly samples an index $z_i$ from $\{1, 2, ..., d\}$ and then reports the sanitized value $\tilde{V}_i^{z_i}$ by Eq.~\ref{eq:rr} to the data collector.
%\begin{align}
%	\label{eq:rr}
%	\tilde{V}_i^{z_i} = \begin{cases}
%		V_i^{z_i}, & \text{w.p.} \quad  \frac{e^\epsilon}{1+e^\epsilon} \\
%		1-V_i^{z_i}, & \text{w.p.}  \quad \frac{1}{1+e^\epsilon} \\
%	\end{cases}
%\end{align}
%Based on the sanitized values from all users, the data collector estimates the count of each value $v \in C_j^k$ (with noise calibration), and derives the roll-up query result by summing up all these counts.
%Note that in this solution, sampling is not compulsory. But if without sampling, each user should perturb each bit of the vector with privacy budget $\frac{\epsilon}{d}$ and then report the whole sanitized vector, which, however, may introduce heavy perturbation~\cite{wang2019collecting}.

{\bf Solution 2: Padding-and-Sampling Perturbation (PSP)}.
Each user first pads (with dummy items) or truncates the item set in her records that belongs to category $c$ into a fixed padding length $\eta \ll |c|$. Then she samples one item from $\eta$ and perturbs it by a categorical-value perturbation mechanism, e.g., $k$-ary Randomized Response (kRR)~\cite{kairouz2014extremal}, Optimized Unary Encoding~\cite{wang2017locally} or Optimal Local Hashing (OLH)~\cite{wang2017locally}. To compensate the effect of sampling, upon receiving the sanitized reports from all users, the data collector aggregates the count and scales it up by a factor of $\eta$. Then the subset counting query $Q(c)$ can be estimated by summing up the counts of all items. This method is adapted from the padding-and-sampling protocol~\cite{wang2018locally}, which has been proved to achieve good performance when the domain size is large.

\subsection{Pitfalls of Existing Solutions}
\label{sec:pitfalls}
NVP intuitively treats the local count as a numerical value and employs a numerical-value perturbation mechanism. However, such perturbation may incur large noise, as the local count can vary a lot among users, which could be as low as $0$ (i.e., a user has no items belonging to a specif category $c$) or as high as $|c|$ (i.e., a user has all items belonging to the category $c$). This inherently results in a large sensitivity for the perturbation mechanism and thus a low utility of subset counting query result.

PSP applies perturbation to a single item instead of the local count on item set, which enables users' reports more informative. However, the effectiveness depends on an appropriate padding length $\eta$. Generally, a large $\eta$ reduces the number of valid items and increases the estimation variance, whereas a small $\eta$ underestimates the subset counting query result~\cite{wang2018locally}. In practice, setting an appropriate value for $\eta$ a priori can be challenging, which hinders its practical application. In addition, for each user it only samples one item and scales it up by $\eta$, which will incur large estimation variance.

To summarize, both solutions suffer from large utility loss due to the noise introduced by heavy perturbation. Additionally, an inappropriate parameter setting further diminishes the utility of PSP.

\section{Randomized Index for Subset Counting}
\label{sec:baseline}
In this section, we present a novel design for answering subset counting query under LDP, namely \underline{C}ounting via \underline{R}andomized \underline{I}ndex (CRI).  We first elaborate on its design rationale, and present a new solution for count aggregation. Then we show the implementation details, followed by comprehensive privacy and utility analysis.

\subsection{Randomized Response vs. Randomized Index}
\label{sec:rr}

In general, given a set of items from each user, a subset counting query returns the total counts of items that belong to a given category. 
%It is noteworthy that this query has the same set-value setting as existing  LDP works~\cite{qin2016heavy,wang2018locally,wang2018privset,ye2019privkv,gu2020pckv}. 
In the literature, all existing works study either frequency estimation of a specific item or top-frequent ones (i.e., heavy hitter identification)~\cite{qin2016heavy, wang2018locally, du2024top}. There is no work in counting a set of items that belong to a category. %applying existing solutions and summing up all value counts would accumulate large noise to the final result. Answering this counting query in non-trivial, and we will present the design rationale of our solution in the next section.

The existing methods randomize a user's real data to ensure ``response-level'' deniability, which inevitably incurs utility loss to the query result. Our key idea is to randomize the indexes of the items being counted to ensure the ``index-level'' deniability, so that the item itself does not need perturbation. The following two examples illustrate the difference between the traditional response-level deniability (Example \uppercase\expandafter{\romannumeral1}) and index-level deniability (Example \uppercase\expandafter{\romannumeral2}).

{\bf Example \uppercase\expandafter{\romannumeral1}.} {\it In a survey there are $10$ sensitive yes/no questions. Each user adopts PSP in Section~\ref{sec:numerical_ldp}. A user first samples one question from them, randomize her true response, and then reports the sanitized response and the question index to which she answers. So the deniability is guaranteed by the randomized response.}

{\bf Example \uppercase\expandafter{\romannumeral2}.} {\it In the same survey, each user randomly samples a question, and then sends her true answer to the data collector, without indicating which question she answers to. So the deniability is guaranteed by the randomized index.}

We observe that Example \uppercase\expandafter{\romannumeral1} exhibits a larger variance and consequently lower accuracy compared to Example \uppercase\expandafter{\romannumeral2}. This is because the former involves both sampling and perturbation error, whereas the latter has the same amount of sampling error but zero perturbation error. This observation motivates us to develop a Counting via Randomized Index (CRI) protocol for answering subset counting queries under the $\epsilon$-LDP guarantee.

\subsection{Overview of CRI Protocol}
As shown in Figure~\ref{fig:overview}, the CRI protocol for subset counting queries consists of three steps. Step \textcircled{\small 1} pre-processes each user's item set by filtering items unrelated to category $c$ and grouping the filtered items. Step \textcircled{\small 2} produces a {\it bit vector} of values in category $c$, each bit corresponding to the existence of one item. Then a bit is sampled from the vector, and sent to the data collector. Upon receiving the sampled bits from all users, the collector estimates the subset counting query result $Q(c)$ in Step \textcircled{\small 3}.

\begin{figure*}
	\centering
	\subfigure{
		\includegraphics[width=\linewidth]{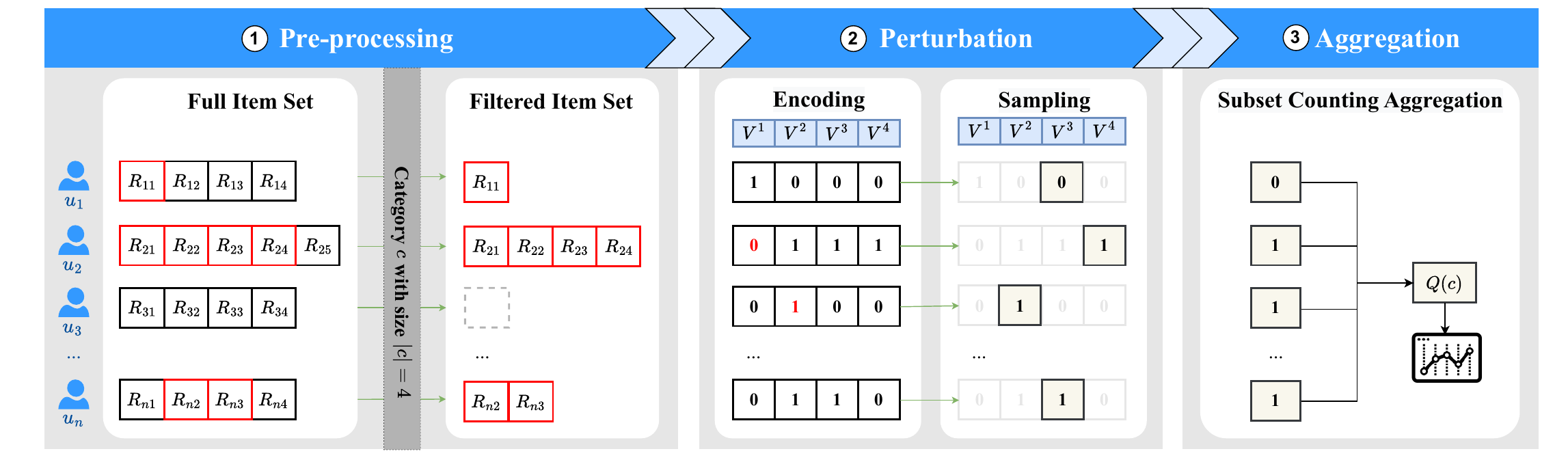}
	}
	
	\caption{Workflow of CRI for Answering Subset Counting Queries.}
	\vspace{-0.1in}
	\label{fig:overview}
\end{figure*}

There remains a privacy issue in the above procedure. Recall that in Figure~\ref{fig:overview}, the subset size of the specified category $c$ is $4$. The user $u_1$ has only one item (i.e., $R_{11}$) belonging to $c$, resulting in the encoded bit vector `1000'. Thus the sampled bit can be either `1' or `0' with some probability. However, for user $u_2$, who has all the items in $c$ and an encoded bit vector of `1111', and user $u_3$, who has none of the items in $c$ and an encoded bit vector of `0000', their sampled bits must be `1' and `0' respectively, i.e., without any deniability. This absolutely violates $\epsilon$-LDP. To address this issue, for those all `1' and all `0' cases, we can randomly flip a bit to ensure there are both bits `1' and `0' in each vector. In Figure~\ref{fig:overview}, the flipped bits are shown in red. 

In what follows, we will elaborate on the CRI protocol, with a focus on its correctness and privacy analysis. 

\subsection{CRI: Counting via Randomized Index}

Recall that given a subset counting query $Q(c)$ on category $c$, each user $u_i$ first filters those unrelated items of $c$ and then encodes her filtered item set into a bit vector $V_i=\{V_i^1, V_i^2, ..., V_i^d\}$ of length $d=|c|$, where each bit $V_i^l$ ($l \in \{1, 2, ..., d\}$) is defined as
\begin{align}
	\label{eq:encoding}
	V^l_i = \begin{cases}
		1, &\text{if} \ \ \exists r \in R_i, \mathbbm{1}_c(r)=1,     \\
		0, &\text{otherwise}.
	\end{cases}
\end{align}
%That means $V^l_i=1$ if the user has the $l$-th value of $C_j^k$, and $V_i^l=0$ otherwise.
Table~\ref{table:cases} shows different cases of bit vectors according to the number of bit `$1$', where $\pi_t$ denotes the proportion of those cases whose number of bit `$1$' is $t$. Note that each case $\pi_t$ involves up to $\tbinom{d}{t}$ different bit vectors. For example, $\pi_1$ is the proportion of vectors `100...0', `010...0', `001...0', ..., `000...1' among all $n$ users.  
Thus, we have $\sum_{t=0}^{d}\pi_t= 1$. %, and each case $\pi_t$ actually includes a set of $\tbinom{d}{t}$ vectors, whose number of bit $1$ is $t$.

\begin{table}[h]
	\caption{Cases of Users' Encoded Bit Vectors}
	\label{table:cases}
	\centering
	\begin{tabular}{|c|c|c|c|c|}
		\hline
		{\bf Proportion} &{\bf \# of $\bm{1}$} &{\bf \# of $\bm{0}$} &{\bf Pr}$\bm{[1]}$ & {\bf Pr}$\bm{[0]}$   \\ \hline
		$\pi_0$ &$0$ &$d$ &$0$ &$1$   \\ \hline
		$\pi_1$ &$1$ &$d-1$ &$1/d$ &$(d-1)/d$   \\ \hline
		$\pi_2$ &$2$ &$d-2$ &$2/d$ &$(d-2)/d$   \\ \hline
		$\pi_3$ &$3$ &$d-3$ &$3/d$ &$(d-3)/d$   \\ \hline
		... &... &... &... &...   \\ \hline
		$\pi_{d-1}$ &$d-1$ &$1$ &$(d-1)/d$ &$1/d$   \\ \hline
		$\pi_d$ &$d$ &$0$ &$1$ &$0$   \\ \hline
	\end{tabular}
\end{table}

Based on the above $d+1$ cases, according to Eq.~\ref{eq:subset_counting_query}, the counting query result on category $c$ is
\begin{align}
	\label{eq:baseline_fre}
	Q(c) = \sum\nolimits_{t=0}^{d} n\pi_t \cdot t
	= n\sum\nolimits_{t=1}^{d} t \pi_t.
\end{align}

In Table~\ref{table:cases}, we also show $\mathrm{Pr}[1]$ (resp. $\mathrm{Pr}[0]$), the probability when a user randomly samples and reports bit `$1$' (resp. `$0$') from the encoded bit vector. Except for the cases of $\pi_0$ and $\pi_d$, all cases have non-zero probabilities to report either `$0$' or `$1$'. To satisfy $\epsilon$-LDP, a random `$0$' (resp. `$1$') should be flipped to `$1$' (resp. `$0$') in the case of $\pi_0$ (resp. $\pi_d$) before sampling. Let $z_i$ denote the sampled index, then the data collector can derive a noisy count from these sampled bits as
\begin{align}
	\label{eq:baseline_fre_observed}
	\bar{\theta} = d\sum\nolimits_{i=1}^{n} V_i^{z_i}.
\end{align}
And its expectation is
\begin{align}
	\label{eq:baseline_fre_observed_expectation1}
		\mathbb{E}[\bar{\theta}] &= n\pi_0 \cdot 1 + \sum\nolimits_{t=1}^{d-1} n\pi_t \cdot t + n\pi_d\cdot(d-1)  \\
	\label{eq:baseline_fre_observed_expectation2}
		&= n \left(
			\sum\nolimits_{t=1}^{d-1} t \pi_t + \pi_0 + (d-1) \pi_d
		\right).
\end{align}

The term $n\pi_0 \cdot 1$ in Eq.~\ref{eq:baseline_fre_observed_expectation1} means there is a bit `1' in the case of $\pi_0$ after flipping, while the term $n\pi_d\cdot(d-1)$ means there are $d-1$ bit `1' in the case of $\pi_d$ after flipping.
The gap between Eqs.~\ref{eq:baseline_fre} and \ref{eq:baseline_fre_observed_expectation2} must be calibrated from $\bar{\theta}$ in Eq.~\ref{eq:baseline_fre_observed} to ensure an unbiased estimation $\tilde{\theta}$:
\begin{align}
	\label{eq:baseline_calibrated}
	\tilde{\theta} = \bar{\theta} + n(\pi_d - \pi_0).
\end{align}

We are yet to derive $\pi_d-\pi_0=\Delta \pi$ in Eq.~\ref{eq:baseline_calibrated}, which is estimated by a privacy budget $\epsilon'$ allocated from the overall budget $\epsilon$. Each user reports a local status flag $y_i$ that indicates whether her case is $\pi_d$, $\pi_0$, or neither of them:
\begin{align}
	\label{eq:pid_pi0_s}
	y_i = \begin{cases}
		1, &\text{if} \ V_i = \{1\}^d, \\
		-1, &\text{if} \ V_i = \{0\}^d,  \\
		0, &\text{otherwise}.
	\end{cases}
\end{align}
The flag is sanitized into $y'_i$ by kRR~\cite{kairouz2014extremal} with the privacy budget $\epsilon'$ : 
\begin{align}
	\label{eq:pid_pi0_grr}
	\textrm{Pr}[y'_i=x] = \begin{cases}
		\frac{e^{\epsilon'}}{2+e^{\epsilon'}}, &\text{if} \ x=y_i, \\
		\frac{1}{2+e^{\epsilon'}}, &\text{if} \ x \in \{1, -1, 0\}\backslash y_i.
	\end{cases}
\end{align}

Upon receiving the sanitized status flags of all users, we can estimate $\Delta \pi$ as
\begin{align}
	\label{eq:pid_pi0}
	\Delta\pi = \frac{(2+e^{\epsilon'})\sum_{i=1}^{n}y'_i}{n(e^{\epsilon'}-1)}.
\end{align}
%where $c'(1)$ and $c'(-1)$ denotes the counts of $1$ and $-1$ respectively among all users' reported status.

Therefore, the subset counting query result on the category $c$ can be estimated as
\begin{align}
	\label{eq:baseline_calibrated_estimate}
	\tilde{Q}(c) = \tilde{\theta} = \bar{\theta} + n\Delta\pi.
\end{align}
In Section~\ref{sec:baseline_analysis}, we will provide the proof of Eq.~\ref{eq:pid_pi0} together with the proof of unbiasedness of Eq.~\ref{eq:baseline_calibrated_estimate}.

\begin{algorithm}
	\footnotesize
	\caption{Workflow of CRI Protocol}
	\begin{tabular}{ll}
		{\bf Input:}  & A category $c$\\
		& All users' item sets $\{R_{i}, R_{i}, ..., R_{n}\}$ \\
		& Privacy budget $\epsilon$ \\
		{\bf Output:} & Estimated subset counting query result $\tilde{Q}(c)$  \\
		{\bf Procedure:} & \\
		\multicolumn{2}{l}{\quad  $//$\textbf{\emph{User side}}}
	\end{tabular}
	\label{alg:baseline}
	\begin{algorithmic}[1]
		\vspace{-0.06in}
		\FOR{each user $u_i \in \mathcal{U}$}
		\STATE Extract the item set $R^*_i$ from $R_i$ with items belonging to $c$ 
		\STATE Encode $R_i^*$ into a bit vector $V_i=\{V_i^1, V_i^2, ..., V_i^d\}$ by Eq.~\ref{eq:encoding}, where $d=|c|$
		\IF{$V_i=\{1\}^d$}
		\STATE Set flag $y_i=1$
		\STATE Randomly flip a bit to $0$
		\ELSIF{$V_i=\{0\}^d$}
		\STATE Set flag $y_i=-1$
		\STATE Randomly flips a bit to $1$
		\ELSE
		\STATE Set flag $y_i=0$
		\ENDIF
		
		\STATE Randomly sample an index $z_i \in \{1, 2, ..., d\}$
		\STATE Perturb $y_i$ to $y'_i$ by Eq.~\ref{eq:pid_pi0_grr} with budget $\epsilon'=\epsilon-\log{(d-1)}$
		\STATE Send $V_i^{z_i}$ and $y'_i$ to the data collector
		\ENDFOR
		
		\vspace{0.05in}
		$//$\textbf{\emph{Collector side}}
		\STATE Calculate the noisy count $\bar{\theta}$ by Eq.~\ref{eq:baseline_fre_observed}
		\STATE Calculate $\Delta\pi$ by Eq.~\ref{eq:pid_pi0}
		\STATE Estimate the counting query result $\tilde{Q}(c)$ by Eq.~\ref{eq:baseline_calibrated_estimate}
		\RETURN Query result $\tilde{Q}(c)$
	\end{algorithmic}
\end{algorithm}

Algorithm~\ref{alg:baseline} summarizes the workflow of CRI protocol for answering a subset counting query. Given a query on category $c$, each user $u_i$ first extracts a filtered record set $R_i^*$ from her original $R_i$ (Line 2) and then encodes the filtered item set $R_i^*$ into a bit vector $V_i$ with length $d=|c|$ (Line 3). Subsequently, the user randomly flips a bit if all bits are $1$ or $0$ (Lines 6 and 9). Meanwhile, the user also sets her local status flag $y_i$ according to Eq.~\ref{eq:pid_pi0_s} (Lines 5, 8 and 11). Then the user randomly samples an index $z_i \in \{1, 2, ..., d\}$ (Line 12) and perturbs her status flag $y_i$ to $y'_i$ by Eq.~\ref{eq:pid_pi0_grr} with privacy budget $\epsilon'$ (Line 13). The computation of $\epsilon'$ will be elaborated in Theorem \ref{thm:baseline_privacy}. Finally, the sampled bit and the sanitized status are sent to the data collector (Line 14). Upon receiving all reports from users, the collector calculates a noisy count $\bar{\theta}$ of bit `$1$' from all sampled bits by Eq.~\ref{eq:baseline_fre_observed}, and calculates $\Delta \pi$ from all status flags by Eq.~\ref{eq:pid_pi0}, and then estimates the counting query result $\tilde{Q}(c)$ (Lines 15-17). %Theorem~\ref{thm:baseline_privacy} in Sec.~\ref{sec:baseline_analysis} shows that this CRI protocol satisfies $\epsilon$-LDP, where $\epsilon=\ln(d-1)+\epsilon'$.

%\subsection{Generalization to Multi-Category Counting Queries}
%\label{sec:rid_plus_extension}
%For now a subset counting query is only defined on a single category. In this subsection, we generalize it to a multi-category setting. For example, {\it ``what is the total counts of items belonging to categories $c_1$ and $c_2$?''} We formally define it as follows. 

%For multi-query setting, as with the paradigm of existing LDP works~\cite{zhang2018calm,ye2019privkv,wang2019collecting}, we can divide the user population into some portions, and let users in each portion answers one query. Specifically, given $q$ roll-up queries, we can uniformly divide the user population (i.e., $n$ users) into $q$ portions, so that $\frac{n}{q}$ users in each portion contributes their data to obtain a roll-up query result.

%\begin{definition}
%({\bf Multi-Category Subset Counting Query.}) Given a set of $T$ categories $c=\{c_1, ..., c_t\}$, a multi-category subset counting query returns the count of items belonging to $c_j$, $\forall j=1,...,t$. Formally,
%\begin{align}
%	\label{eq:query}
%	Q(c) = \sum_{i=1}^{n} \sum_{r\in R_i} \sum_{j=1}^{t} \mathbbm{1}(r, c_j)
%\end{align}
%\end{definition}

%In essence, we can still treat this query as a single-category after ``expanding" the dimensions in these categories. That is, we get all items that belong to these categories, and then treat them as the item domain for the single-category subset counting query.

\subsection{Privacy and Utility Analysis}
\label{sec:baseline_analysis}

In this subsection, we establish the privacy and utility guarantee of our CRI protocol for subset counting query. In particular, Theorem~\ref{thm:baseline_privacy} proves that Algorithm~\ref{alg:baseline} satisfies $\epsilon$-LDP. Theorem~\ref{thm:baseline_unbiased} ensures the estimated result is unbiased, and Theorems~\ref{thm:baseline_variance} and \ref{thm:baseline_variance} provide the error bound of the estimation variance.

\begin{theorem}
	\label{thm:baseline_privacy}
Algorithm~\ref{alg:baseline} satisfies $\epsilon$-LDP, where $\epsilon=\ln(d-1) + \epsilon'$, $d$ is the size of query category, and $\epsilon'$ is the privacy budget for perturbing the status flag by Eq.~\ref{eq:pid_pi0_grr}.
\end{theorem}
\begin{proof}
	For a specific category $c$, each user sends a sampled bit and her sanitized status flag to the data collector. For the bit sampling, we know that each user may sample a bit '$1$' or '$0$'. From Table~\ref{table:cases}, the highest probability to sample a bit $1$ (or $0$) is $\frac{d-1}{d}$, while the lowest probability is $\frac{1}{d}$. Therefore, for any two users with filtered item sets $R^*_i$ and $R^*_j$ regarding to category $c$, and for any sampled bit $b \in \{0, 1\}$, we have
	\begin{align*}
		\frac{\mathrm{Pr}[\textrm{CRI}(R^*_i)=b]}{\mathrm{Pr}[\textrm{CRI}(R^*_j)=b]}
		\le \frac{(d-1)/d}{1/d} = e^{\ln(d-1)}.
	\end{align*}
	Therefore, sampling a bit by CRI satisfies $\ln(d-1)$-LDP. On the other hand, for any status $y'$ reported by CRI, we know from Eq.~\ref{eq:pid_pi0_grr} that
	\begin{align*}
		\frac{\mathrm{Pr}[\textrm{CRI}(R^*_i)=y']}{\mathrm{Pr}[\textrm{CRI}(R^*_j)=y']}
		\le \frac{e^{\epsilon'}/(2+e^{\epsilon'})}{1/(2+e^{\epsilon'})} = e^{\epsilon'}.
	\end{align*}
It means reporting the status flag by CRI satisfies $\epsilon'$-LDP. Then according to the sequential composition in Theorem~\ref{theorem:composition}, Algorithm~\ref{alg:baseline} satisfies $\epsilon$-LDP, where $\epsilon=\ln(d-1)+\epsilon'$.
\end{proof}

\begin{theorem}
	\label{thm:baseline_unbiased}
	The estimated counting query result on any category $c$ by Eq.~\ref{eq:baseline_calibrated_estimate} is unbiased, i.e., $\mathbb{E}[\tilde{Q}(c)] = Q(c)$.
\end{theorem}
\begin{proof}
	%Our CTM framework disassembles a roll-up query $Q(C_j^k)$ into the multiplication of a counting query and a mean query, i.e., $Q(C_j^k)=C(C_j^k) \times M(C_j^k)$. For the estimated mean $\tilde{M}({C}_j^k)$, it is unbiased when the perturbation mechanism is unbiased, i.e., $\mathbb{E}[\tilde{M}({C}_j^k)] = M({C}_j^k)$. 
  As for the count estimation, according to Eq.~\ref{eq:pid_pi0_grr}, we know when $y_i=1$, $\mathbb{E}[y'_i]=\frac{e^{\epsilon'}}{2+e^{\epsilon'}}\cdot 1 + \frac{1}{2+e^{\epsilon'}}\cdot (-1) + \frac{1}{2+e^{\epsilon'}}\cdot 0 = \frac{e^{\epsilon'}-1}{2+e^{\epsilon'}}$. Similarly, when $y_i=-1$, $\mathbb{E}[y'_i]=\frac{1-e^{\epsilon'}}{2+e^{\epsilon'}}$, and when $y_i=0$, $\mathbb{E}[y'_i]=0$. Therefore,
  %As for the count estimation, according to Eq.~\ref{eq:pid_pi0_grr}, we know when $y_i=1$, $\mathbb{E}[y'_i]=\frac{e^{\epsilon'}-1}{2+e^{\epsilon'}}$, and when $y_i=-1$, $\mathbb{E}[y'_i]=\frac{1-e^{\epsilon'}}{2+e^{\epsilon'}}$, and when $y_i=0$, $\mathbb{E}[y'_i]=0$. Therefore,
	\begin{align*}
		\sum_{i=1}^{n} \mathbb{E}[y'_i]
		&= c_1 \cdot \frac{e^{\epsilon'}-1}{2+e^{\epsilon'}}
			+ c_{-1} \cdot \frac{1-e^{\epsilon'}}{2+e^{\epsilon'}} \\
		&= \frac{(c_1 - c_{-1})(e^{\epsilon'}-1)}{2+e^{\epsilon'}},
	\end{align*}
	where $c_1$ and $c_{-1}$ are the real counts of status flags $1$ and $-1$ respectively among all users. Then by Eq.~\ref{eq:pid_pi0}, we have
	\begin{align*}
		\mathbb{E}[\Delta \pi]
%		&= \frac{(2+e^{\epsilon'})\mathbb{E}[\sum_{i=1}^{n}y'_i]}{n(e^{\epsilon'}-1)} \\
		&= \frac{(2+e^{\epsilon'})\sum_{i=1}^{n}\mathbb{E}[y'_i]}{n(e^{\epsilon'}-1)}
%		&= \frac{
%				(2+e^{\epsilon'})\frac{(c_1 - c_{-1})(e^{\epsilon'}-1)}{2+e^{\epsilon'}}
%			}{n(e^{\epsilon'}-1)} \\
		= \frac{c_1-c_{-1}}{n} \\
		&= \frac{\sum_{i=1}^{n} \mathbbm{1}(V_i=\{1\}^d)}{n}
			- \frac{\sum_{i=1}^{n} \mathbbm{1}(V_i=\{0\}^d)}{n} \\
		&= \pi_d - \pi_0,
	\end{align*}
	which mean $\Delta \pi$ by Eq.~\ref{eq:pid_pi0} is an unbiased estimation of $\pi_d-\pi_0$.
	By substituting the above $\mathbb{E}[\Delta \pi]$ and $\mathbb{E}[\tilde{\theta}]$ in Eq.~\ref{eq:baseline_fre_observed_expectation2} to Eq.~\ref{eq:baseline_calibrated_estimate}, we have
	\begin{align*}
		&\quad \mathbb{E}[\tilde{Q}(c)]
		= \mathbb{E}[\bar{\theta}] + n(\mathbb{E}[\Delta\pi])  \\
		&= n \left(
				\sum\nolimits_{t=1}^{d-1} t \pi_t + \pi_0 + (d-1) \pi_d
			\right)  + n (\pi_d - \pi_0) \\
		&=  n\sum\nolimits_{t=1}^{d} t \pi_t
		= Q(c).
	\end{align*}
	As such, the estimation of the subset counting query by Eq.~\ref{eq:baseline_calibrated_estimate} is unbiased. 
\end{proof}

\begin{theorem}
	\label{thm:baseline_variance}
	Given a category $c$, the number of users $n$, and privacy budget $\epsilon'$ for status flag perturbation, the estimation variance of the counting query result by CRI in Algorithm~\ref{alg:baseline} is bounded by
	$\frac{1}{4}nd^2 + \frac{2n(e^{\epsilon'}+2)}{(e^{\epsilon'}-1)^2}$.
\end{theorem}
\begin{proof}
	According to Eq.~\ref{eq:baseline_fre_observed},
	\begin{align*}
		\mathrm{Var}[\bar{\theta}] &= d^2 \cdot \mathrm{Var}[\sum\nolimits_{i=1}^{n} V_i^{z_i}]
		= d^2 \cdot \sum\nolimits_{i=1}^{n} \mathrm{Var}[V_i^{z_i}]  \\
		&= n d^2 \left(
			(\pi_0+\pi_d)\frac{d-1}{d^2} + \sum\nolimits_{t=1}^{d-1} \pi_t \frac{t(d-t)}{d^2}
		\right) \\
		&\le nd^2 \left( \frac{\pi_0+\pi_d}{4} + \sum\nolimits_{t=1}^{d-1} \frac{\pi_t}{4} \right) \\
		&= \frac{1}{4}nd^2.
	\end{align*}
	
	Let $c_1$ and $c_{-1}$ denote the real counts of status flags $1$ and $-1$ respectively, and $c'_1$ and $c'_{-1}$ denote observed counts based on users' reports. Then we know
	\begin{align*}
		& \mathrm{Var}[c'_1]
		= \frac{2e^{\epsilon'} \cdot c_1}{(2+e^{\epsilon'})^2}
			+ \frac{(1+e^{\epsilon'}) (n-c_1)}{(2+e^{\epsilon'})^2},  \\
		& \mathrm{Var}[c'_{-1}]
		=  \frac{2e^{\epsilon'} \cdot c_{-1}}{(2+e^{\epsilon'})^2}
			+ \frac{(1+e^{\epsilon'})(n-c_{-1})}{(2+e^{\epsilon'})^2},  \\
		& \mathrm{Cov}[c'_1, c'_{-1}] = -\frac{e^{\epsilon'} (c_1+c_{-1})}{(2+e^{\epsilon'})^2}
			- \frac{n-c_1-c_{-1}}{(2+e^{\epsilon'})^2}.
	\end{align*}
	
	Therefore,
	\begin{align*}
		\mathrm{Var}[\sum_{i=1}^{n}y'_i] &= \mathrm{Var}[c'_1-c'_{-1}]  \\
		&= \mathrm{Var}[c'_1] + \mathrm{Var}[c'_{-1}] - 2\mathrm{Cov}[c'_1, c'_{-1}] \\
		&= \frac{n(1+5e^{\epsilon'})+3(n-c_1-c_{-1})(1-e^{\epsilon'})}{(2+e^{\epsilon'})^2}.
	\end{align*}
	
	According to Eqs.~\ref{eq:pid_pi0_grr} and \ref{eq:pid_pi0},
	\begin{align*}
		\mathrm{Var}[\Delta\pi] &= \frac{
			(2+e^{\epsilon'})^2 \cdot \mathrm{Var}[\sum_{i=1}^{n}y'_i]
		}{n^2(e^{\epsilon'}-1)^2} \\
		&= \frac{
			n(1+5e^{\epsilon'})+3(n-c_1-c_{-1})(1-e^{\epsilon'})
		}{n^2(e^{\epsilon'}-1)^2} \\
		&\le \frac{
			n(1+5e^{\epsilon'})+3n(1-e^{\epsilon'})
		}{n^2(e^{\epsilon'}-1)^2} \\
		&= \frac{
			2(e^{\epsilon'}+2)
		}{n(e^{\epsilon'}-1)^2}.
	\end{align*}
	
	According to Eq.~\ref{eq:baseline_calibrated_estimate}, we have
	\begin{align*}
		\mathrm{Var}[\tilde{Q}(c)]
		&= \mathrm{Var}[\bar{\theta}] + n^2 \cdot \mathrm{Var}[\Delta\pi] \\
		&\le \frac{1}{4}nd^2 + \frac{2n(e^{\epsilon'}+2)}{(e^{\epsilon'}-1)^2},
	\end{align*}
	which proves that the variance of frequency estimation by CRI is bounded by
	$\frac{1}{4}nd^2 + \frac{2n(e^{\epsilon'}+2)}{(e^{\epsilon'}-1)^2}$.
\end{proof}

%More specifically, we show sampling error and perturbation error of RID in Fig.~\ref{fig:rid_error}, by varying the given privacy budget $\epsilon'$ for perturbing each user's status $s_i$.
%\begin{figure}[h]
%	\centering
%	\subfigure{
%		\includegraphics[width=0.8\linewidth]{figs/sampling_perturbation_er-crop.pdf}
%	}
%	\caption{Sampling and Perturbation Error of RID}
%	\label{fig:rid_error}
%\end{figure}

By Theorem~\ref{thm:baseline_variance}, we observe that the estimation error of subset counting query by CRI comes from two sources, namely the sampling and perturbation process. While ensuring deniability for extreme cases where bits are all `$1$'s or `$0$'s, the perturbation comes at a price of an estimation variance of $\frac{2n(e^{\epsilon'}+2)}{(e^{\epsilon'}-1)^2}$. In other words, to derive an estimation of $\Delta\pi$ in Eq.~\ref{eq:pid_pi0} and thus make the counting query result unbiased, {\bf CRI sacrifices its utility by re-introducing the perturbation error}. In the next section, we find an alternative way to ensure deniability for extreme cases, and propose CRIAD which eliminates the perturbation error and thus enhances the utility.

\section{CRIAD: Counting via Randomized Index with Augmented Dummies}
\label{sec:rid}
In this section, we present a utility-enhanced CRI solution to counting queries. The main idea is to augment a user's encoded bit vector with dummy bits, so the protocol is called \underline{C}ounting via \underline{R}andomized \underline{I}ndex with \underline{A}ugmented \underline{D}ummies (CRIAD). In this section, we first present the skeleton of CRIAD, followed by its customization to cope with various privacy requirements and category sizes. Finally, we summarize its overall procedure, together with the privacy and utility analysis.

\subsection{Randomized Index with Augmented Dummies}
\label{sec:rid_plus_dummy}

To ensure the deniability of two extreme cases (all `$0$'s and `$1$'s), we augment a user's bit vector with dummy $0/1$ bits, so that either bit can be sampled in both extreme cases and therefore perturbation is no longer needed. Table~\ref{table:cases_01pair} illustrates this effect when a bit `$1$' is added to each bit vector, so in the case of $\pi_0$, both {\bf Pr}$\bm{[1]}$ and {\bf Pr}$\bm{[0]}$ are non-zero.  %, and the overall procedure (except the case of $\pi_d$) satisfies $\ln d$-LDP, as the ratio of any two sampling probabilities for reporting a bit $1$ (or $0$) is bounded by $\frac{d/(d+2)}{1/(d+2)} = e^{\ln d}$.

\begin{table}[h]
	\scriptsize
	\caption{Cases of Bit Vectors with an Augmented Bit `1'}
	\label{table:cases_01pair}
	\centering
	%	\begin{tabular}{|c|c|c|c|c|}
	\begin{tabular}{|@{ }p{1.0cm}<{\centering}|@{ }p{1.0cm}<{\centering}|@{ }p{1.0cm}<{\centering}|@{ }p{1.8cm}<{\centering}|@{ }p{1.7cm}<{\centering}|}
		\hline
		{\bf Proportion} &{\bf No. of $\bm{1}$} &{\bf No. of $\bm{0}$} &{\bf Pr}$\bm{[1]}$ & {\bf Pr}$\bm{[0]}$   \\ \hline
		$\pi_0$ &$1$ &$d$ &$1/(d+1)$ &$d/(d+1)$   \\ \hline
		$\pi_1$ &$2$ &$d-1$ &$2/(d+1)$ &$(d-1)/(d+1)$   \\ \hline
		$\pi_2$ &$3$ &$d-2$ &$3/(d+1)$ &$(d-2)/(d+1)$   \\ \hline
		$\pi_3$ &$4$ &$d-3$ &$4/(d+1)$ &$(d-3)/(d+1)$   \\ \hline
		... &... &... &... &...   \\ \hline
		$\pi_{d-1}$ &$d$ &$1$ &$d/(d+1)$ &$1/(d+1)$   \\ \hline
		$\pi_d$ &$d+1$ &$0$ &$1$ &$0$   \\ \hline
	\end{tabular}
\end{table}

Apparently we can add another dummy bit `$0$' to each bit vector to fix the case of $\pi_d$ as well. Nonetheless, these dummies come at a price of sampling error, because they dilute the original bit distribution in these vectors. For example, in the $\pi_1$ case, the sampling probability of bit `$1$' changes from $\frac{1}{d}$ (Table~\ref{table:cases}) to $\frac{2}{d+1}$ (Table~\ref{table:cases_01pair}), and will further change to  $\frac{3}{d+2}$ if $2$ dummies are added. This obviously increases the sampling variance. Our key idea is that in real-world applications, the $\pi_d$ case (i.e., all bits are `$1$'s) is too rare to contribute to the overall count. This is especially true when the category size is large. Therefore, we can just suppress a $\pi_d$ cases to a $\pi_{d-1}$ case by randomly flipping one bit $1$ to $0$ to refrain from adding a dummy bit `$0$'. %As long as the proportion $\pi_d$ is small enough or even approach $0$, the suppression effect on the utility can be omitted.

\begin{algorithm}
	\footnotesize
	\caption{Procedure of CRIAD with One Dummy Bit}
	\label{alg:rid_1dummy}
	\begin{tabular}{ll}
		{\bf Input:} & A category $c$\\
		& All users' item sets $\{R_{i}, R_{i}, ..., R_{n}\}$ \\
		{\bf Output:} & A sampled bit $V_i^{z_i}$ \\
		{\bf Procedure:} & \\
		\multicolumn{2}{l}{\quad  $//$\textbf{\emph{User side}}}
	\end{tabular}
	\begin{algorithmic}[1]
		\vspace{-0.06in}
		\FOR{each user $u_i \in \mathcal{U}$}
			\STATE Extract the item set $R^*_i$ from $R_i$ with items belonging to $c$ 
			\STATE Encode $R_i^*$ into a binary vector $V_i=\{V_i^1, V_i^2, ..., V_i^d, 1\}$ by Eq.~\ref{eq:encoding}, where $d=|c|$
			\IF{$V_i=\{1\}^{d+1}$}
				\STATE Randomly flip a bit to $0$
			\ENDIF
		
			\STATE Randomly sample an index $z_i \in \{1, 2, ..., d+1\}$
			\STATE Report $V_i^{z_i}$ to the data collector 
		\ENDFOR
		
		\vspace{0.05in}
		$//$\textbf{\emph{Collector side}}
		\STATE Estimate the subset counting query result $\tilde{Q}(c)$ by Eq.~\ref{eq:rid_plus_1dummy}
		\RETURN Query result $\tilde{Q}(c)$
	\end{algorithmic}
\end{algorithm}

Algorithm~\ref{alg:rid_1dummy} shows the pseudo-code of the above procedure, where one dummy bit `$1$' is appended as the $(d+1)$-th bit in each user's bit vector (Line 3). If all bits in $V_i$ are `$1$'s, the user randomly flips one of them to $0$ (Lines 4-5). Then each user randomly samples an index $z_i \in \{1, 2, ..., d+1\}$ and reports the sampled bit $V_i^{z_i}$ to the data collector (Lines 6-7). At the collector side, the impact of added dummy bits `$1$' on the counting query can be eliminated by subtracting $n$ from the aggregated bits, as each of $n$ user contributes a dummy bit `$1$',
\begin{align}
	\label{eq:rid_plus_1dummy}
	\tilde{Q}(c) = (d+1)\sum\nolimits_{i=1}^{n} V_i^{z_i}-n,
\end{align}
where $V_i^{z_i}$ is the reported bit from user $u_i$. The following two theorems establish the privacy and correctness guarantee of Algorithm~\ref{alg:rid_1dummy}, respectively.

%From Theorems~\ref{thm:rid_plus_01pair} and \ref{thm:baseline_variance}, we observe that, by adding two dummies (i.e., a 01 pair), perturbation variance is eliminated while the sampling variance is scaled up by a factor $(1+\frac{2}{d})^2$. The root cause of the latter lies in the increasing domain size from $d$ to $d+2$. In other words, adding more dummies result in the increasing sampling variance, and vice versa. For instance, if we only add a bit $1$ as a dummy, then the variance can be reduced to $\frac{1}{4} n (d+1)^2$. However, this destroys the deniability of those cases with all bits $1$. To address this issue, similar to RID, we can still let the user, whose encoded vector are all bit $1$, randomly flips a bit to $0$. But instead, we don't intend to estimate the proportion of these cases among the user population, as we believe it is nearly close to $0$ in many real-world applications, where the domain size is usually large enough. For example, a roll-up operation for ``mobile apps'' involves a domain size over 1M, and it is almost impossible that a user has installed all of these apps. In a word, this is a practical assumption in real-world applications, which helps to reduce the estimation variance.

\begin{theorem}
	\label{thm:rid_plus_1dummy}
	Algorithm~\ref{alg:rid_1dummy} satisfies $\ln d$-LDP.
\end{theorem}
\begin{proof}
	By Algorithm~\ref{alg:rid_1dummy}, the case of $\pi_d$ is reduced to $\pi_{d-1}$ by randomly flipping a bit '$1$' to '$0$'. So for any two filtered item sets $R^*_i$ and $R^*_j$ regarding category $c$ from users, and any bit $b\in \{0, 1\}$ reported, we know
	\begin{align*}
		\frac{\mathrm{Pr}[\textrm{CRI}(R^*_i)=b]}{\mathrm{Pr}[\textrm{CRI}(R^*_j)=b]}
		& \le \frac{
			\mathrm{Pr}[b=1|V_i\in \pi_{d-1}]
		}{
			\mathrm{Pr}[b=1|V_j\in \pi_0]
		} \\
		&= \frac{d/(d+1)}{1/(d+1)}
		= e^{\ln d}.
	\end{align*}
	Therefore, Algorithm~\ref{alg:rid_1dummy} satisfies $\ln d$-LDP.
\end{proof}

\begin{theorem}
	\label{theorem:unbias_rid_plus}
	The estimated counting query result $\tilde{Q}(c)$ by Eq.~\ref{eq:rid_plus_1dummy} is unbiased if each user's number of bit '$1$' in the bit vector does not exceed $d-1$, and the estimation variance is bounded by $\frac{1}{4} n (d+1)^2$.
\end{theorem}
\begin{proof}
	According to Eq.~\ref{eq:rid_plus_1dummy}, the expectation of the subset counting query result is 	
	\begin{align*}
		\mathbb{E}[\tilde{Q}(c)]
		&= (d+1)\cdot \mathbb{E}[\sum\nolimits_{i=1}^{n} V_i^{z_i}] - n  \\
		&= (d+1) \left(
		 	\sum\nolimits_{t=0}^{d-1} \frac{t+1}{d+1} \cdot n\pi_t + \frac{d}{d+1} n\pi_d
		 \right)   -  n \\
%		&= (d+1) \left(
%			\sum\nolimits_{t=0}^{d} \frac{t+1}{d+1} \cdot n\pi_t - \frac{1}{d+1} n\pi_d
%		\right)   -  n \\
		&= n \left(
			\sum\nolimits_{t=1}^{d}t \pi_t + \sum\nolimits_{t=0}^{d}\pi_t - \pi_d
		\right) - n  \\
		&= Q(c) - n \pi_d.
	\end{align*}
	Since the number of bit `$1$' does not exceed $d-1$, $\pi_d=0$. Therefore, $\mathbb{E}[\tilde{Q}(c)] = Q(c)$, i.e., $\tilde{Q}(c)$ is unbiased.
	
	As for the estimation variance, we have
	\begin{align*}
		\mathrm{Var}[\tilde{Q}(c)]
		&= (d+1)^2 \mathrm{Var}[\sum\nolimits_{i=1}^{n} V_i^{z_i}] \\
		&= (d+1)^2 \sum\nolimits_{t=0}^{d} n\pi_t \frac{(t+1)(d-t)}{(d+1)^2}  \\
		&\le  (d+1)^2\sum\nolimits_{t=0}^{d}  \frac{n\pi_t}{4}  \\
		&= \frac{1}{4} n (d+1)^2.
	\end{align*}
	Therefore, the estimation variance of the subset counting query is bounded by $\frac{1}{4} n (d+1)^2$.
\end{proof}

%Obviously, Algorithm~\ref{alg:rid_plus} satisfies $\ln d$-LDP. Similar to Theorem~\ref{thm:rid_plus_01pair}, we also know that the estimation $\bar{f}_j^k$ by Algorithm~\ref{alg:rid_plus} is unbiased when the number of bits $1$ of the encoded vector does not exceed $d-1$, and the estimation variance is bounded by $\frac{1}{4n}(1+\frac{1}{d})^2$.

\subsection{Customizing CRIAD}
\label{sec:rid_plus_scalability}
CRIAD is generic in terms of the number of dummies and samples in each user. In this subsection, we show the customization of CRIAD to support a wide range of privacy requirements and category sizes.

\subsubsection{Multiple Dummies}
\label{sec:rid_plus_multi_dummy}
In Algorithm~\ref{alg:rid_1dummy}, only one dummy bit `$1$' is added to each user's bit vector. Table~\ref{table:cases_multi_dummy} shows the impact of $m$ dummies on the sampling probabilities of bits `$1$'s and `$0$'s respectively. We observe that for the first $d+1-m$ cases (i.e., from $\pi_0$ to $\pi_{d-m}$), the sampling probability of bit `$1$' (resp. $0$) ranges from $\frac{m}{d+m}$ to $\frac{d}{d+m}$ (resp. from $\frac{d}{d+m}$ to $\frac{m}{d+m}$). As more dummy `$1$'s are added, the sampling probability gradually approaches $1$ (resp. $0$). This motivates us to confine the number of bit `$1$'s to $d-m$. As such, the probability ratio of any two sampled bits can be bounded by $d/m$ and thus $m$ dummies can satisfy ($ln\frac{d}{m}$)-LDP (see Theorem~\ref{thm:rid_plus_privacy} for complete proof).%Similarly, this can achieved by letting each user randomly flip some bits to $0$ so that the number of bits $0$ is at least $m$.

\begin{table}[h]
	\scriptsize
	\caption{Cases of Bit Vectors with $m$ Dummy Bits `1'}
	\label{table:cases_multi_dummy}
	\centering
	\begin{tabular}{|@{ }p{1.0cm}<{\centering}|@{ }p{0.98cm}<{\centering}|@{ }p{0.9cm}<{\centering}|@{ }p{2.17cm}<{\centering}|@{ }p{1.81cm}<{\centering}|}
		\hline
		{\bf Proportion} &{\bf No. of $\bm{1}$} &{\bf No. of $\bm{0}$} &{\bf Pr}$\bm{[1]}$ & {\bf Pr}$\bm{[0]}$   \\ \hline
		$\pi_0$ &$m$ &$d$ &$m/(d+m)$ &$d/(d+m)$   \\ \hline
		$\pi_1$ &$1+m$ &$d-1$ &$(1+m)/(d+m)$ &$(d-1)/(d+m)$   \\ \hline
		$\pi_2$ &$2+m$ &$d-2$ &$(2+m)/(d+m)$ &$(d-2)/(d+m)$   \\ \hline
		... &... &... &... &...   \\ \hline
		$\pi_{d-m}$ &$d$ &$m$ &$d/(d+m)$ &$m/(d+m)$   \\ \hline
		$\pi_{d-m+1}$ &$d+1$ &$m-1$ &$(d+1)/(d+m)$ &$(m-1)/(d+m)$   \\ \hline
		... &... &... &... &...   \\ \hline
		$\pi_{d-1}$ &$d+m-1$ &$1$ &$(d+m-1)/(d+m)$ &$1/(d+m)$   \\ \hline
		$\pi_d$ &$d+m$ &$0$ &$1$ &$0$   \\ \hline
	\end{tabular}
\end{table}

Then each user $u_i$ randomly samples an index $z_i \in \{1, 2, ..., d+m\}$ and reports the bit $V_i^{z_i}$ to the data collector. Based on the reports from all users, the estimated subset counting query result $\tilde{Q}(c)$ can be derived as
\begin{align*}
	\tilde{Q}(c) = (d+m)\sum\nolimits_{i=1}^{n} V_i^{z_i}-mn,
\end{align*}
where the second term $mn$ is the number of dummy `$1$'s added by all users.

\subsubsection{Multiple Samples}
In Algorithm~\ref{alg:rid_1dummy}, each user randomly samples and reports one bit to the data collector, and the overall algorithm satisfies $\ln d$-LDP (see Theorem~\ref{thm:rid_plus_1dummy}). However, this becomes a privacy bottleneck when the privacy budget $\epsilon>\ln d$, as the extra budget has to be wasted. CRIAD can benefit from a large privacy budget by having users report multiple samples. Note that this is different from directly applying sequential composition (i.e., Theorem~\ref{theorem:composition}) to repeatedly perform one-bit sampling multiple times, as in CRIAD, bits are sampled without replacement to cover as many data bits as possible. Table~\ref{table:cases_multi_dummies_samples} shows the impact of number of samples $s$ on the probabilities of bits `$1$'s and `$0$'s respectively, where $m$ ($m\ge s$) dummies are added. Note that the table only shows the first $d-m+1$ cases (i.e., from $\pi_0$ to $\pi_{d-m}$), as the others are reduced to the case of $\pi_{d-m}$ before sampling, in the same way as in Section~\ref{sec:rid_plus_multi_dummy}.

\renewcommand\arraystretch{1.3}
\begin{table}[h]
	\scriptsize
	\caption{Cases of Bit Vectors with $m$ Dummy Bits `1' and $s$ Samples}
	\label{table:cases_multi_dummies_samples}
	\centering
	\begin{tabular}{|@{ }p{1.0cm}<{\centering}|@{ }p{0.98cm}<{\centering}|@{ }p{0.9cm}<{\centering}|@{ }p{2.17cm}<{\centering}|@{ }p{1.81cm}<{\centering}|}
		\hline
		{\bf Proportion} &{\bf No. of $\bm{1}$} &{\bf No. of $\bm{0}$} &{\bf Pr}$\bm{[\{1\}^s]}$ & {\bf Pr}$\bm{[\{0\}^s]}$   \\ \hline
		$\pi_0$ &$m$ &$d$ &$\tbinom{m}{s}\big/\tbinom{d+m}{s}$ &$\tbinom{d}{s}\big/\tbinom{d+m}{s}$   \\ \hline
		$\pi_1$ &$1+m$ &$d-1$ &$\tbinom{m+1}{s}\big/\tbinom{d+m}{s}$ &$\tbinom{d-1}{s}\big/\tbinom{d+m}{s}$   \\ \hline
		$\pi_2$ &$2+m$ &$d-2$ &$\tbinom{m+2}{s}\big/\tbinom{d+m}{s}$ &$\tbinom{d-2}{s}\big/\tbinom{d+m}{s}$   \\ \hline
		... &... &... &... &...   \\ \hline
		$\pi_{d-m}$ &$d$ &$m$ &$\tbinom{d}{s}\big/\tbinom{d+m}{s}$ &$\tbinom{m}{s}\big/\tbinom{d+m}{s}$   \\ \hline
	\end{tabular}
\end{table}
\renewcommand\arraystretch{1.0}

%From Table~\ref{table:cases_multi_dummies_samples}, we know that, with $m$ dummies and $s$ samples from each user, the overall procedure satisfies $\ln\big(\tbinom{d}{s}\big/\tbinom{m}{s}\big)$-LDP.
Let $z_i=\{z_i[1], z_i[2], ..., z_i[s]\}$ denote the $s$ indexes sampled by user $u_i$. Based on all users' reports, the estimated counting query result $\tilde{Q}(c)$ can be derived as
\begin{align*}
	\tilde{Q}(c) = \frac{d+m}{s}\sum\nolimits_{i=1}^{n} \sum\nolimits_{x=1}^{s}V_i^{z_i[x]}-mn.
\end{align*}

\subsubsection{Multiple Groups}
Although $m$ dummies can satisfy  ($ln\frac{d}{m}$)-LDP, when the category size $d$ is large, it is difficult to satisfy a small privacy budget. To address this issue, we further propose a grouping strategy to divide a large category over $\{1, 2, ..., d\}$ into $g$ disjoint and equal-sized groups $\{G_1, G_2, ..., G_g\}$, i.e., $|G_r|=\frac{d}{g}$ and $\cup_{r=1}^{g}G_r=\{1, 2, ..., d\}$. Each group becomes a new (sub)category and thus the above multi-dummy and multi-sample strategies can still work in each group. The users are also divided into $g$ equal-sized groups $\{U_1, U_2, ..., U_g\}$, i.e., $\frac{n}{g}$ users for each group, and each user reports $s$ samples drawn from her bit vector with $m$ dummy `$1$'s in her corresponding group.

Upon receiving reports from all users, the data collector first counts bit `$1$'s in each group, and then collectively derives the estimated counting query result based on the counts from $g$ groups. Specifically, for group $G_r$, the count in user group $U_r$ can be estimated as
\begin{align*}
	\gamma_r &= \frac{d+gm}{gs}\sum_{i=1}^{|U_r|}\sum_{x=1}^{s}V_{U_r[i]}^{z_i[x]} - m\cdot |U_r|,
\end{align*}
where $U_r[i]$ denotes the $i$-th user in $U_r$, and
$\sum_{i=1}^{|U_r|}\sum_{x=1}^{s}V_{U_r[i]}^{z_i[x]}$ is the sum of all returned samples in group $U_r$. Finally, the estimated counting query result becomes
\begin{align}
	\label{eq:rid_plus_frequency_all}
	\tilde{Q}(c) = \sum_{r=1}^{g} \gamma_r \cdot g = \frac{d+gm}{s}\sum_{i=1}^{n}\sum_{x=1}^{s}V_i^{z_i[x]} - nmg.
\end{align}

\subsection{CRIAD: Putting Things Together}
\label{sec:rid_plus_together}

Algorithm~\ref{alg:criad} shows the complete CRIAD procedure with a multi-dummy, multi-sample, and multi-group strategy. As such, Algorithm~\ref{alg:rid_1dummy} can be considered as a special case where $m=s=g=1$. Given a subset counting query on the category $c$, the data collector first derives three parameters, namely, the number of dummies $m$, samples $s$ and groups $g$, based on the given privacy budget $\epsilon$ and category size $d=|c|$ (Line 1), which will be elaborated by Theorem~\ref{thm:rid_plus_parameter} in Section~\ref{sec:rid_plus_analysis}. 
The collector then broadcasts $m$, $s$ and $g$ to all users (Line 2). 
At the user side, the domain $\{1, 2, ..., d\}$ of category $c$ is first divided into $g$ groups $\{G_1, G_2, ..., G_g\}$ uniformly at random (Line 3), then each user samples a group $G_r$ for reporting (Line 5). 
Each user extracts a filtered record set $R_i^*$ from $R_i$ with items belonging to $G_r$ (Line 6), and then encodes it into a bit vector with length $|G_r|$ (Line 7). Then $m$ dummy bits $1$ are added to the encoded bit vector (Line 8). If the number of bit `0's is fewer than $m$, the user needs to randomly flip some `1's to ensure at least $m$ `0's (Lines 9-11). Then $s$ indexes are randomly sampled from $\{1, 2, ..., d+m\}$ and the user sends these sampled bits to the data collector (Lines 12-13). Finally, based on all the reports from users, the collector estimates the subset counting query result by Eq.~\ref{eq:rid_plus_frequency_all}.

\begin{algorithm}[t]
	\footnotesize
	\caption{Workflow of CRIAD}
	\begin{tabular}{ll}
		{\bf Input:} & A category $c$\\
					& All users' item set $\{R_{i}, R_{i}, ..., R_{n}\}$ \\
					 & Privacy budgets $\epsilon_1$ and $\epsilon_2$ for count and mean estimation \\
		{\bf Output:} & Estimated subset counting query result $\tilde{Q}(c)$  \\
		{\bf Procedure:} & \\
		\multicolumn{2}{l}{\quad  $//$\textbf{\emph{Collector side}}}
	\end{tabular}
	\label{alg:criad}
	\begin{algorithmic}[1]
		\vspace{-0.06in}
		\STATE Set parameters: $m, s, g \leftarrow ParaSelect(d,\epsilon)$ by Theorem~\ref{thm:rid_plus_parameter}		
		
		\STATE Broadcast $m$, $s$ and $g$ to all users
		
		\vspace{0.05in}
		$//$\textbf{\emph{User side}}
		\STATE Divide the full domain $\{1, 2, ..., d\}$ of category $c$ into $g$ groups $\{G_1, G_2, ..., G_g\}$ uniformly at random
		\FOR{each user $u_i$ ($1 \le i \le n$)}
			\STATE Randomly sample a group $G_r$ for $r\in \{1, 2, ..., g\}$
			\STATE Extract the item set $R^*_i$ from $R_i$ with items belonging to $G_r$ 
			\STATE Encode $R^*_i$  into a binary vector $V_i=\{0, 1\}^{|G_r|}$ by Eq.~\ref{eq:encoding}
			\STATE Add $m$ dummy bits (i.e., $\{1\}^m$) to $V_i$
			\STATE Set $c'$ as the number of bit `0's in $V_i$
			\IF{$c'<m$}
				\STATE Randomly flip $m-c'$ bits `1' in $V_i$
			\ENDIF
			
			\STATE Randomly sample $s$ indices $z_i = \{z_i[1], ..., z_i[s]\}$ from $\{1, 2, ..., d+m\}$
			\STATE Send $V_i^{z_i}$ to the data collector
		\ENDFOR
			
		\vspace{0.05in}
		$//$\textbf{\emph{Collector side}}
		\STATE Estimate the counting query result $\tilde{Q}(c)$ by Eq.~\ref{eq:rid_plus_frequency_all}
		\RETURN Query result $\tilde{Q}(c)$
	\end{algorithmic}
\end{algorithm}

\subsection{Privacy and Utility Analysis of CRIAD}
\label{sec:rid_plus_analysis}

In this subsection, we will address the pending problem of choosing parameters $m$, $s$ and $g$, given category $c$ and privacy budget $\epsilon$. We will first provide privacy and utility analysis in Theorems~\ref{thm:rid_plus_privacy} to \ref{thm:rid_plus_variance}, based on which we derive the optimal setting for three parameters in Theorem~\ref{thm:rid_plus_parameter}.

\begin{theorem}
	\label{thm:rid_plus_privacy}
	With $m$ dummies, $s$ samples and $g$ groups, CRIAD satisfies $\ln \left(\tbinom{d/g}{s} \big/ \tbinom{m}{s}\right)$-LDP.
\end{theorem}
\begin{proof}
	%Note that Algorithm~\ref{alg:rid_plus} is a special case of RID$^+$ where $m=s=g=1$.
	Recall that in Table~\ref{table:cases_multi_dummy}, the last $m$ cases (i.e., $\pi_{d-m+1}, \pi_{d-m+2}, ..., \pi_d$) are reduced to $\pi_{d-m}$ by suppressing the number of bit `$1$'s in the encoded bit vector. Therefore, for any two filtered item sets $R^*_i$ and $R^*_j$ regarding a subset, and any bit $b\in \{0, 1\}$ reported by CRIAD, we have
	\begin{align*}
		\frac{\mathrm{Pr}[\textrm{CRIAD}(R^*_i)=b]}{\mathrm{Pr}[\textrm{CRIAD}(R^*_j)=b]}
		& \le \frac{
			\mathrm{Pr}[\textrm{CRIAD}(V_i\in \pi_{d-m})=1]
		}{
			\mathrm{Pr}[\textrm{CRIAD}(V_j\in \pi_0)=1]
		} \\
		&= \frac{d/(d+m)}{m/(d+m)}
		= \frac{d}{m}.
	\end{align*}

	Then with increasing $s>1$, let $\bm{b}=\{0,1\}^s$ denote a bit vector of length $s$ reported by a user. Recall that in Table~\ref{table:cases_multi_dummies_samples}, the above ratio further becomes
	\begin{align*}
		\frac{\mathrm{Pr}[\textrm{CRIAD}(R^*_i)=\bm{b}]}{\mathrm{Pr}[\textrm{CRIAD}(R^*_j)=\bm{b}]}
		& \le \frac{
			\mathrm{Pr}[\textrm{CRIAD}(V_i\in \pi_{d-m})=\{1\}^s]
		}{
			\mathrm{Pr}[\textrm{CRIAD}(V_j\in \pi_0)=\{1\}^s]
		} \\
		&= \frac{\tbinom{d}{s}/\tbinom{d+m}{s}}{\tbinom{m}{s}/\tbinom{d+m}{s}}
		= \frac{\tbinom{d}{s}}{\tbinom{m}{s}}.
	\end{align*}

	Then with increasing $g>1$, the group size changes from $d+m$ to $\frac{d}{g}+m$. So the above ratio becomes $\tbinom{d/g}{s} \big/ \tbinom{m}{s}$. Therefore, CRIAD satisfies $\epsilon$-LDP, where $\epsilon = \ln \left(\tbinom{d/g}{s} \big/ \tbinom{m}{s}\right)$.
\end{proof}

\begin{theorem}
	\label{thm:rid_plus_unbiasedness}
	The estimated count $\tilde{Q}(c)$ by Algorithm~\ref{alg:criad} is unbiased if the number of bits `$1$'s in the encoded bit vector does not exceed $d/g-m$. 
\end{theorem}
\begin{proof}
	By the grouping strategy, each encoded bit vector is split into $g$ sub-vectors. Therefore, for an encoded vector $V_i$ which contains $t$ bit `$1$'s, the length of each sub-vector is $d/g$ and the expected count of bit $1$ is $t/g$. 	
	In Eq.~\ref{eq:rid_plus_frequency_all}, $\sum_{x=1}^{s}V_i^{z_i[x]}$ is the count of bit `$1$'s in $s$ samples reported by the user $u_i$ in a group. Similar to Theorem~\ref{theorem:unbias_rid_plus}, if the number of bit `$1$'s in that group does not exceed $d/g-m$, the count $\sum_{x=1}^{s}V_i^{z_i[x]}$ from each user is an unbiased estimation of the true count in that group. Hence, in the case of $\pi_t$, the expectation of the count $\sum_{x=1}^{s}V_i^{z_i[x]}$ is 	
	
	\begin{align*}
		\mathbb{E}[\sum\nolimits_{x=1}^{s}V_i^{z_i[x]}|\pi_t]
		&= \frac{s(t/g+m)}{d/g+m}.
 	\end{align*}
 	Therefore,
	 \begin{align*}
	 	\mathbb{E}[\sum\nolimits_{x=1}^{s}V_i^{z_i[x]}]
	 	&= \sum\nolimits_{t=0}^{d} \pi_t \cdot \mathbb{E}[\sum\nolimits_{x=1}^{s}V_i^{z_i[x]}|\pi_t] \\
	 	&= \sum\nolimits_{t=0}^{d} \frac{s(t/g+m)}{d/g+m} \pi_t.
	 \end{align*}

	By Eq.~\ref{eq:rid_plus_frequency_all}, the expectation of the estimated count is 	
	\begin{align*}
		\mathbb{E}[\tilde{Q}(c)]
		&= \frac{d+gm}{s} \sum\nolimits_{i=1}^{n} \mathbb{E}[\sum\nolimits_{x=1}^{s}V_i^{z_i[x]}] - nmg \\
		&= \frac{d+gm}{s} \sum\nolimits_{i=1}^{n} \sum\nolimits_{t=0}^{d} \frac{s(t/g+m)}{d/g+m} \pi_t  - nmg \\
	%	&= ng\cdot \sum\nolimits_{t=0}^{d} (\frac{t}{g}+m) \pi_t  -nmg  \\
	%	&= ng\cdot \sum\nolimits_{t=0}^{d} \frac{t}{g} \pi_t + nmg\sum\nolimits_{t=0}^{d}\pi_t  - nmg \\
		&= n\sum\nolimits_{t=0}^{d} t \pi_t
		= Q(c).
	\end{align*}

	Therefore, $\mathbb{E}[\tilde{Q}(c)] = Q(c)$.		
\end{proof}

\begin{theorem}
	\label{thm:rid_plus_bias}
	The expected bias error for $\tilde{Q}(c)$ by Algorithm~\ref{alg:criad} is $n\sum_{t=d-mg}^{d} \pi_tf(t)$, where $f(t)=t-d+mg$.		
\end{theorem}

\begin{proof}
	By the grouping strategy, each encoded bit vector is split into $g$ sub-vectors. Therefore, for an encoded vector $V_i$ which contains $t$ bit `$1$'s, the length of each sub-vector is $d/g$ and the expected count of bit $1$ is $t/g$. 
	
	For the case where $t>d-mg$, since we suppress $\pi_{t/g}$ cases to $\pi_{d/g-m}$ case by randomly flipping one bit $1$ to $0$ to refrain from adding a dummy bit `$0$', the expected bias error after aggregation is
	\begin{align*}
	 n\sum\nolimits_{t=d-mg}^{d}\pi_t (\frac{t}{g}-\frac{d}{g}+m)g=n\sum\nolimits_{t=d-mg}^{d}\pi_tf(t),
	 \end{align*}
	 where $f(t)=t-d+mg$. Conversely, when $t\le d-mg$, the expected bias error after aggregation is $0$. In summary, the expected bias error for $\tilde{Q}(c)$ by Algorithm~\ref{alg:criad} is $n\sum_{t=d-mg}^{d} \pi_t f(t)$.
\end{proof}

\begin{theorem}
	\label{thm:rid_plus_variance}
	With $m$ dummies, $s$ samples and $g$ groups, the variance of the estimated count by CRIAD is bounded by $\frac{n(d+gm)^2}{4 s}$.
\end{theorem}
\begin{proof}
	With $g$ groups, $m$ dummies and $s$ samples, each user first selects a group uniformly at random, and then selects $s$ sample bits in that group. From the perspective of sampling variance, this is equivalent to directly selecting $s$ samples from the whole domain of length $d+gm$. For any encoded bit vector belonging to the case of $\pi_t \in\{\pi_0, \pi_1, ..., \pi_{d-m}\}$, the whole domain consists of $t+gm$ bit `$1$'s and $d-t$ bit `$0$'s. As such, the variance of the estimated count $\tilde{Q}(c)$ by Eq.~\ref{eq:rid_plus_frequency_all} is
	\begin{align*}
		\mathrm{Var}[\tilde{Q}(c)]
		&= 	\frac{(d+gm)^2}{s^2} \mathrm{Var}[\sum\nolimits_{i=1}^{n} \sum\nolimits_{x=1}^{s}V_i^{z_i[x]}] \\
		&=	\frac{n(d+gm)^2}{s^2} \sum\nolimits_{t=0}^{d} \pi_t \mathrm{Var}[\sum\nolimits_{x=1}^{s}V_i^{z_i[x]} | \pi_t]  \\
		&\lesssim \frac{n(d+gm)^2}{s} \cdot \frac{(t+gm)(d-t)}{(d+gm)^2}  \\
		&< \frac{n(d+gm)^2}{4 s}.
	\end{align*}
	
	Therefore, the variance of the estimated count by CRIAD is bounded by $\frac{n(d+gm)^2}{4 s}$.
%	Note that $\sum_{x=1}^{s}V_i^{z_i[x]}$ means the count of bit $1$ in $s$ samples reported by the user $u_i$. In particular, let $c$ denote $\sum_{x=1}^{s}V_i^{z_i[x]} \in \{0, 1, ..., s\}$, which follows a multinomial distribution. For any $c \in \{0, 1, ..., s\}$, and in any case of $\pi_t \in\{\pi_0, \pi_1, ..., \pi_{d-m}\}$ which consists of $t+gm$ bit $1$ and $d-t$ bit $0$, we have
%	\begin{align*}
%		\mathrm{Pr}[\sum_{x=1}^{s}V_i^{z_i[x]}=c|\pi_t]
%		= \frac{\tbinom{t+gm}{c} \tbinom{d-t}{s-c}}{\tbinom{d+gm}{s}} = p^c_t
%	\end{align*}
%	Then the variance of $\sum_{x=1}^{s}V_i^{z_i[x]}$ in any case $\pi_t$ can be represented as
%	\begin{align*}
%		\mathrm{Var}[\sum_{x=1}^{s}V_i^{z_i[x]} | \pi_t]
%		&= \sum_{c=0}^{s} p_t^c \cdot c^2 - \left(\sum_{c=0}^{s}p_t^c \cdot c\right)^2 \\
%		&= \sum_{c=1}^{s} p_t^c \cdot c^2 - \frac{s^2(t+gm)^2}{(d+gm)^2} \\
%		&< s^2 \left( 1 - \frac{(t+gm)^2}{(d+gm)^2}\right) \\
%		&\le s^2 \left( 1 - \frac{g^2 m^2}{(d+gm)^2}\right)
%	\end{align*}
%	
%	Therefore, the variance of $\bar{f}_j^k$ is
%	\begin{align*}
%		\mathrm{Var}[\bar{f}_j^k]
%		&= 	\frac{(d+gm)^2}{n^2 d^2 s^2} \mathrm{Var}[\sum_{i=1}^{n} 		\sum_{x=1}^{s}V_i^{z_i[x]}] \\
%		&=	\frac{(d+gm)^2}{n^2 d^2 s^2} \sum_{i=1}^{n} \mathrm{Var}[\sum_{x=1}^{s}V_i^{z_i[x]}] \\
%		&=	\frac{(d+gm)^2}{n d^2 s^2} \sum_{t=0}^{d} \pi_t \mathrm{Var}[\sum_{x=1}^{s}V_i^{z_i[x]} | \pi_t] \\
%		&<	\frac{(d+gm)^2}{n d^2 s^2} \cdot
%			s^2 \left( 1 - \frac{g^2 m^2}{(d+gm)^2}\right) \\
%		&= \frac{d+2gm}{nd}
%	\end{align*}
\end{proof}

%Theorem~\ref{thm:rid_plus_variance} derives an upper bound of the estimation variance. In real applications, we don't have to always assume the worst case. As the estimation variance also depends on real data distribution to some extent, we further derive a close-to-real estimation variance by assuming
Finally, Theorem~\ref{thm:rid_plus_parameter} below derives an optimal approximation of $m$, $s$ and $g$ in terms of the estimation variance. 

\begin{theorem}
	\label{thm:rid_plus_parameter}	
	Given privacy budget $\epsilon$, the optimal $m$, $s$, and $g\in \mathbb{Z}^{+}$ of CRIAD can be approximated by

	\begin{align}
		\label{eq:parameter_rid_plus}
	    &\quad m,s,g = \mathop{\arg\min}_{m,s,g}\mathbb{E}[(\tilde{Q}(c)-Q(c))^2] \nonumber \\
	    &=  \mathop{\arg\min}_{m,s,g} \! \left(  n\frac{(d\!+\!gm)^2}{4s}+(n \! \sum_{t=d-mg}^{d} \! \pi_t f(t))^2\right)
	\end{align}

	\begin{align}
		s.t., \quad & \epsilon \ge  \ln \left(\tbinom{d/g}{s} \big/ \tbinom{m}{s}\right) ,\nonumber \\
		& 1<s\le m \le d/g.  \nonumber
		%& 2g\le d  ,\nonumber \\
		%& m, s, g \in \mathbb{Z}^{+} \nonumber
	\end{align}
	%where $n$ is the user population size, and $d$ is the domain size of the category concerned in roll-up operation.
\end{theorem}
\begin{proof}	
	The optimal parameter setting of $m$, $s$ and $g$ can be derived by minimizing the expected squared error, i.e., $\mathop{\arg\min}_{m,s,g}\mathbb{E}[(\tilde{Q}(c)-Q(c))^2]$.
%	\begin{align}
%		&\quad \mathop{\arg\min}_{m,s,g}\mathbb{E}[(\tilde{Q}(c)-Q(c))^2]\\
%		&= \nonumber 
%		\mathop{\arg\min}_{m,s,g} \left( n\frac{(d+gm)^2}{4s}+( n\sum_{t=d-mg}^{d}\pi_tf(t)) ^2\right) , \nonumber
%	\end{align}
    Note that  the first term in Eq.~\ref{eq:parameter_rid_plus} is the overall variance derived from Theorem~\ref{thm:rid_plus_variance}, and the second term is the expected bias error due to suppression derived from Theorem~\ref{thm:rid_plus_bias}. As for the two conditions, the first is due to Theorem~\ref{thm:rid_plus_privacy}, where CRIAD satisfies $\ln \left(\tbinom{d/g}{s} \big/ \tbinom{m}{s}\right)$-LDP. The second is because even when the true bit vector has all bit `$0$'s, the number of bit `$1$' is still $m$ as $m$ dummy `$1$'s are added. And $m\le d/g$ is because in each group we confine the number of bit `$1$'s to $d/g-m$ .
\end{proof}

As for $\pi_t$, we can randomly allocate a portion of users to estimate the distribution, where each user only needs to report an integer. 
Then we conduct a greedy search to enumerate all combinations of $m$, $s$ and $g$, calculate their resulted variance by Eq. 15, and select the optimal one with the lowest variance. Note that $m, s, g\in \mathbb{Z}^{+}$, and $s$ and $g$ are typically small integers even for a large domain. As such, the number of combinations will not be too large, resulting in an efficient search.

%we may conduct a greedy search to enumerate all possible $m$, $s$ and $g$ and then find the optimal ones, when the domain size $d$ is not very large (e.g., $d<10^3$). For a larger $d$,

\section{Experimental Evaluation}
\label{sec:experiment}
In this section, we evaluate the performance of our proposed solution CRIAD to validate its effectiveness in answering subset counting queries.

\begin{figure*}[ht]
	\centering
	\begin{minipage}{0.55\linewidth}
		\centerline{\includegraphics[width=\linewidth]{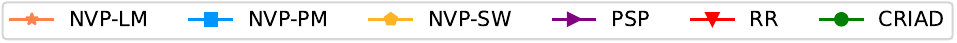}}
		\vspace{0.05in}
	\end{minipage}
	
	\begin{minipage}{0.28\linewidth}
		\centerline{\includegraphics[width=\linewidth]{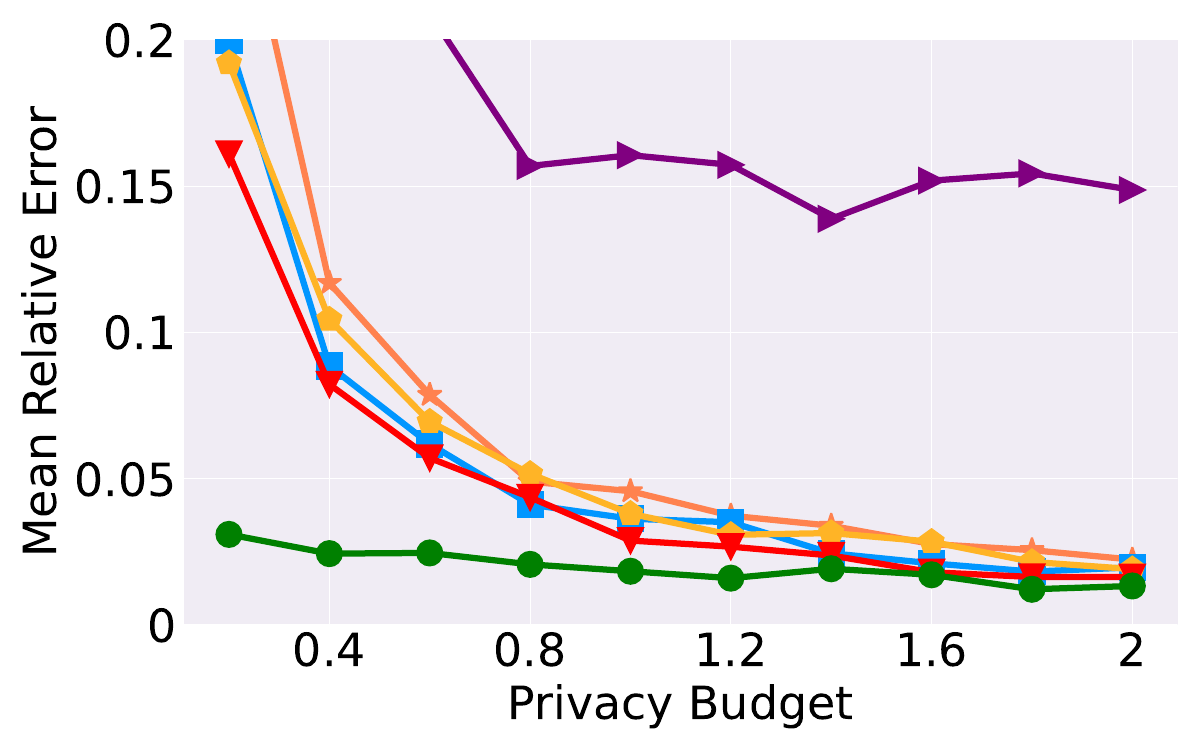}}
		\centerline{\footnotesize{(a) {\it Kosarak}, $[1,100]$}}
	\end{minipage}
	\begin{minipage}{0.28\linewidth}
		\centerline{\includegraphics[width=\linewidth]{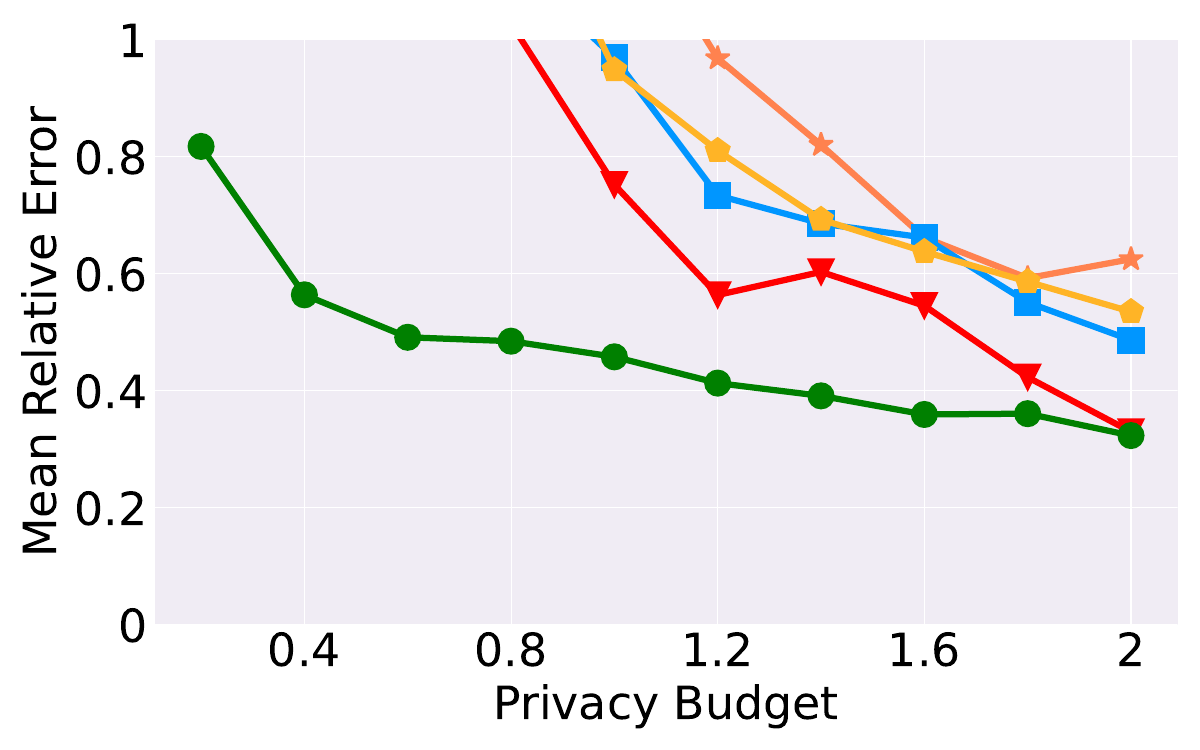}}
		\centerline{\footnotesize{(b) {\it OnlineRetail}, $[1,100]$}}
	\end{minipage}
	\begin{minipage}{0.28\linewidth}
		\centerline{\includegraphics[width=\linewidth]{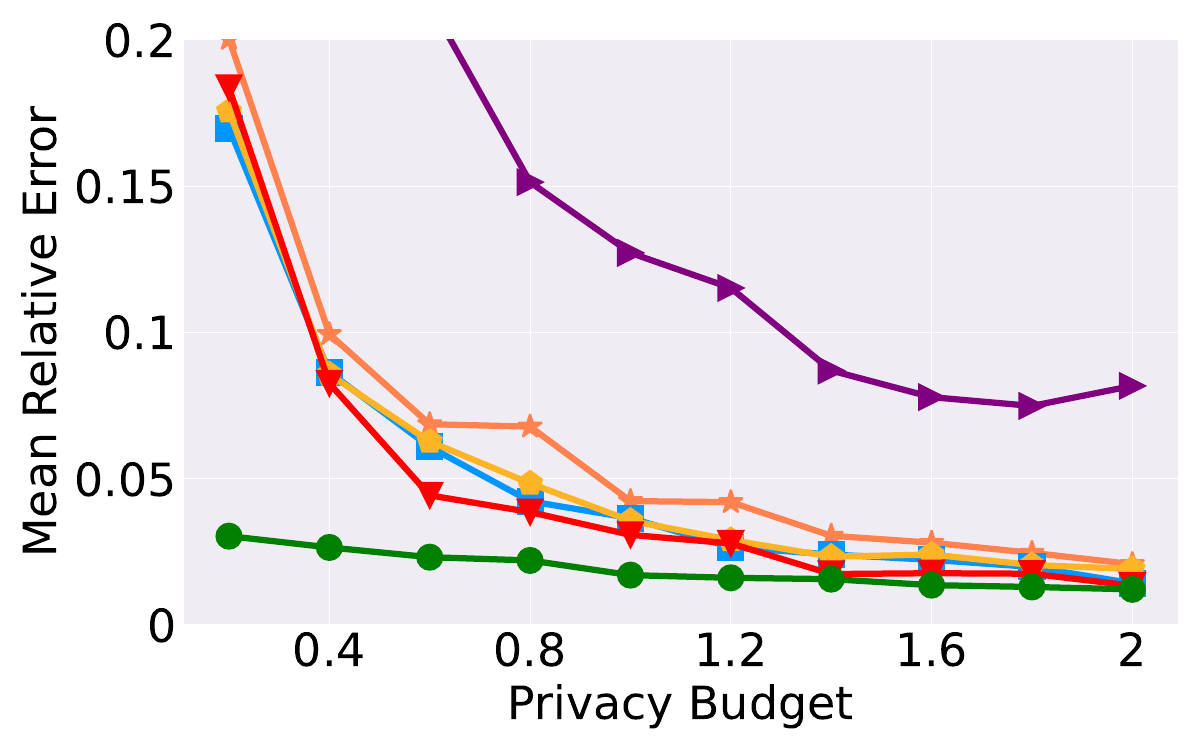}}
		\centerline{\footnotesize{(c) {\it POS}, $[1,100]$}}
	\end{minipage}

	\begin{minipage}{0.28\linewidth}
		\centerline{\includegraphics[width=\linewidth]{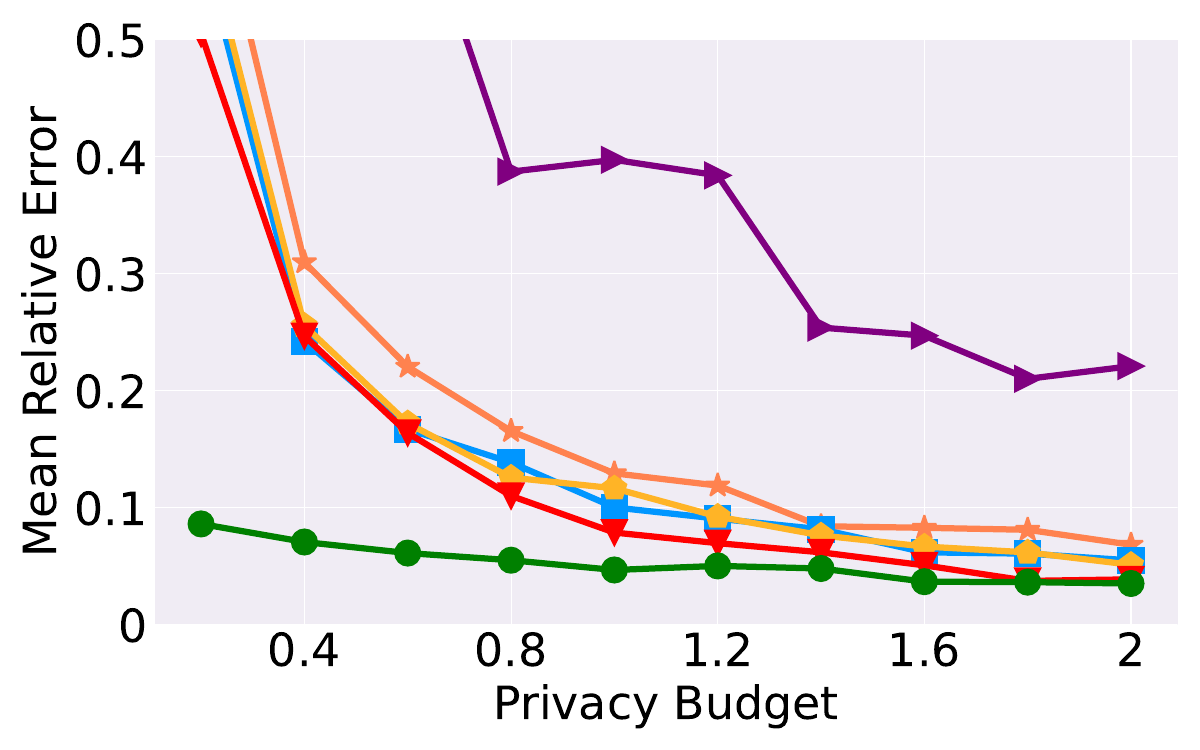}}
		\centerline{\footnotesize{(d) {\it Kosarak}, $[1,400]$}}
	\end{minipage}
	\begin{minipage}{0.28\linewidth}
		\centerline{\includegraphics[width=\linewidth]{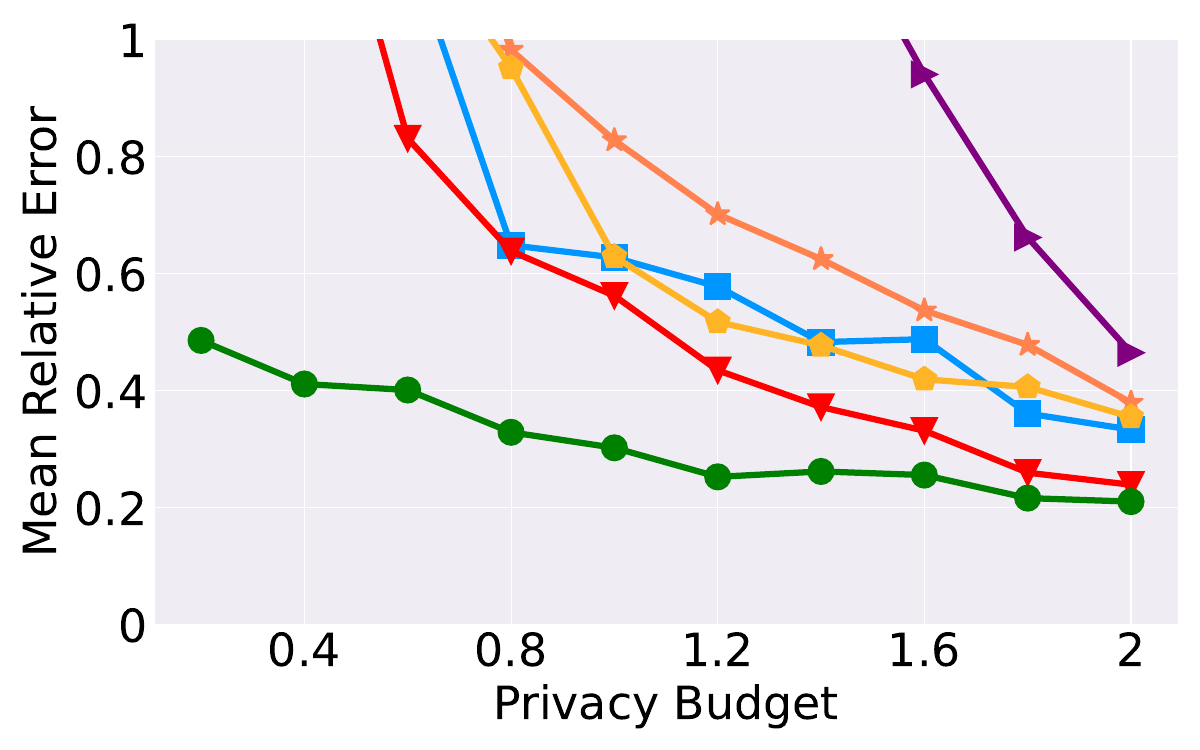}}
		\centerline{\footnotesize{(e) {\it OnlineRetail}, $[1,400]$}}
	\end{minipage}
	\begin{minipage}{0.28\linewidth}
		\centerline{\includegraphics[width=\linewidth]{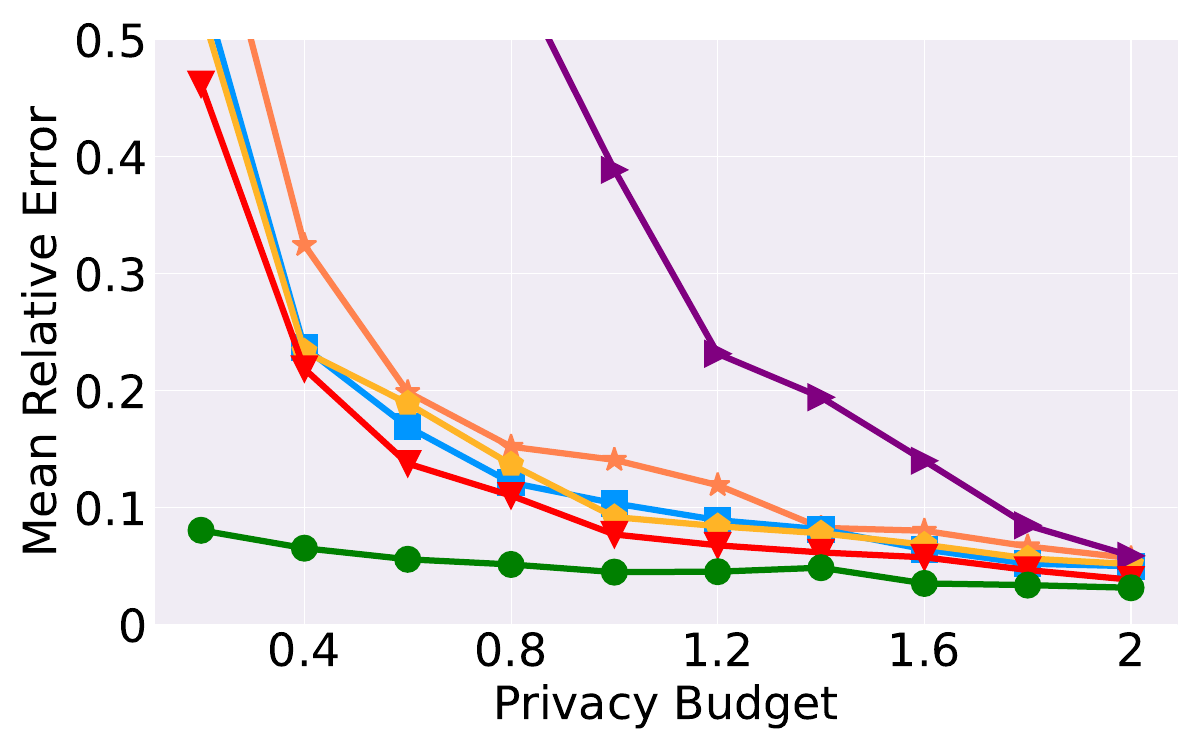}}
		\centerline{\footnotesize{(f) {\it POS}, $[1,400]$}}
	\end{minipage}
	
	\begin{minipage}{0.28\linewidth}
		\centerline{\includegraphics[width=\linewidth]{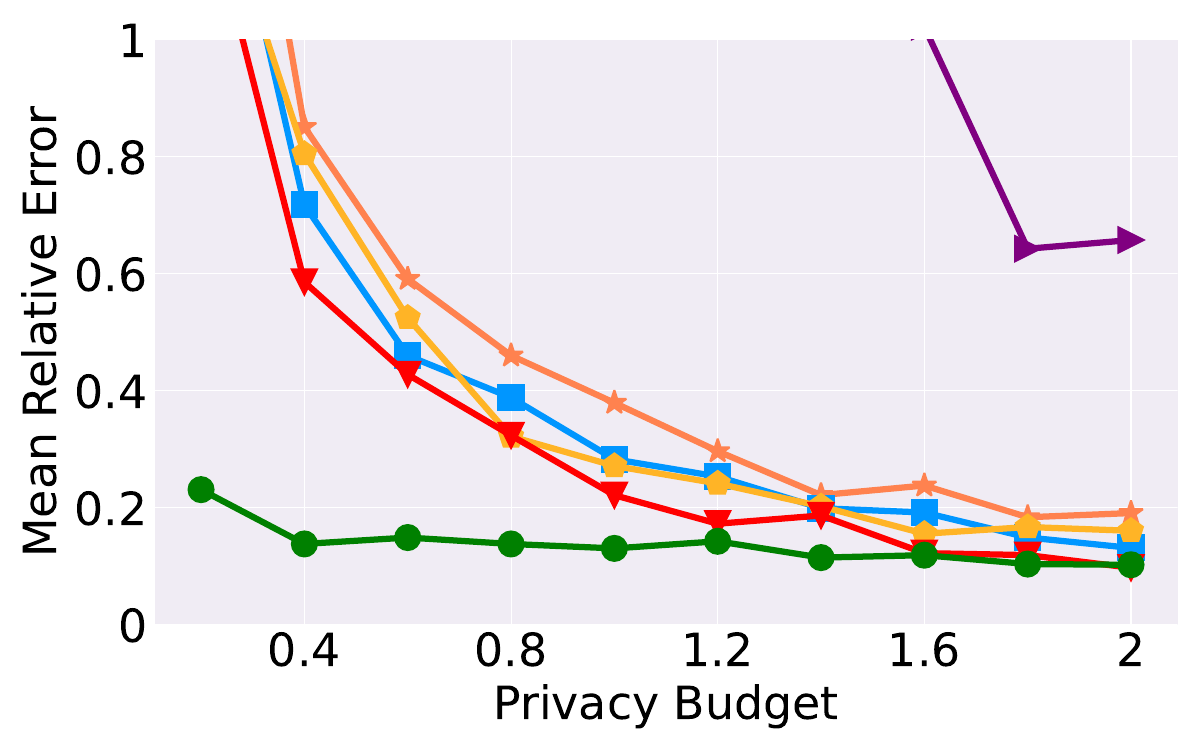}}
		\centerline{\footnotesize{(g) {\it Kosarak}, $[1,1600]$}}
	\end{minipage}
	\begin{minipage}{0.28\linewidth}
		\centerline{\includegraphics[width=\linewidth]{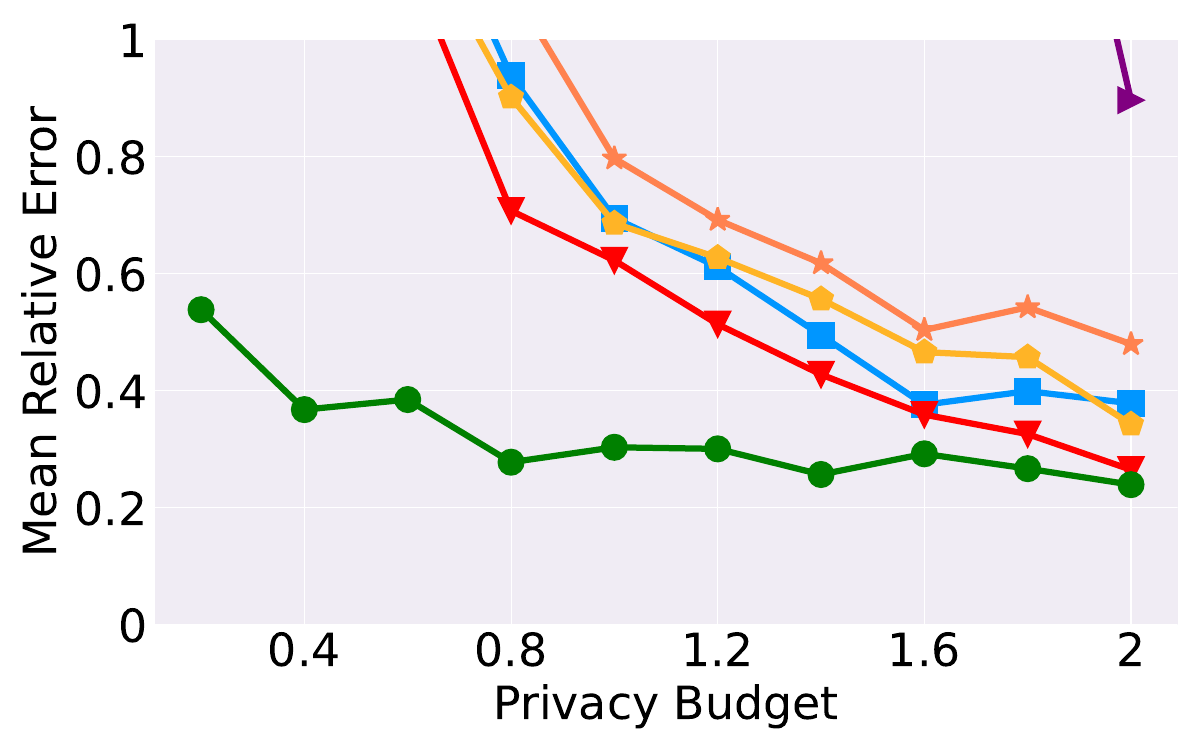}}
		\centerline{\footnotesize{(h) {\it OnlineRetail}, $[1,1600]$}}
	\end{minipage}
	\begin{minipage}{0.28\linewidth}
		\centerline{\includegraphics[width=\linewidth]{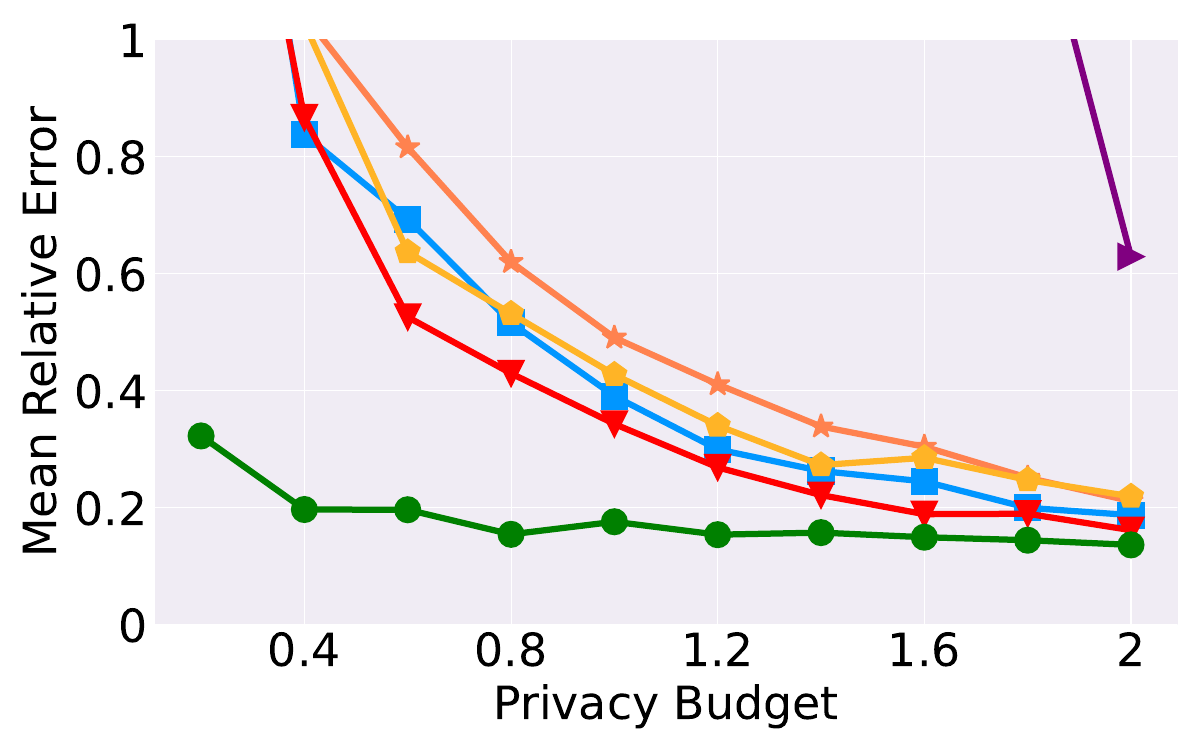}}
		\centerline{\footnotesize{(i) {\it POS}, $[1,1600]$}}
	\end{minipage}
	\caption{Overall performance on real-world datasets with varying privacy budgets.}
	%\vspace{-0.1in}
	\label{fig:varying privacy budget}
\end{figure*}

\subsection{Experiment Setup}

%\subsubsection{Datasets}
\noindent
{\bf Datasets.} We conduct experiments on three real datasets.
\begin{itemize}
	\item
	{\it Kosarak}~\footnote{http://fimi.uantwerpen.be/data/}
	contains click-stream dataset from a Hungarian on-line news portal, which contains $990,002$ users. The item domain size is $41,270$. 
	
	\item
	{\it OnlineRetail}~\footnote{https://archive.ics.uci.edu/ml/datasets/}  
	 is transformed from the Online Retail dataset, which contains $541,909$ users. The item domain is $2,603$. 
	
	\item 
	{\it POS}~\cite{zheng2001real} is a dataset on merchant transactions, which contains $515,595$ users. The item domain is $1,657$.
	
\end{itemize}

\noindent
{\bf Competitors.} 
We compare our proposed solution \underline{CRIAD} with existing value-perturbation LDP solutions, namely Numerical Value Perturbation (\underline{NVP}) and Padding-and-Sampling Perturbation (\underline{PSP}) introduced in Section~\ref{sec:numerical_ldp}. Additionally, to demonstrate the superiority of randomized index over randomized response (\underline{RR}), we also implement RR, which replaces Step \textcircled{\small 2} of Figure~\ref{fig:overview} by sampling a bit from the encoded vector produced by Eq.~\ref{eq:encoding}, perturbing it by RR~\cite{warner1965randomized}, and reporting the sanitized bit for aggregation. 

\noindent
{\bf Parameter Setting.} For NVP, we implement it by integrating three state-of-the-art perturbation mechanisms, namely Laplace Mechanism (LM), Piecewise Mechanism (PM) and Square Wave mechanism (SW), which are denoted by NVP-LM, NVP-PM, and NVP-SW, respectively. For PSP, we follow the existing work~\cite{wang2018locally} and use the 90th percentile of the users' itemset sizes of the query category as the padding length $\eta$. To estimate this value, $10\%$ of users report the length through Optimal Local Hashing (OLH)~\cite{wang2017locally} in advance. For CRIAD, to derive the optimal parameter setting of $m$, $s$ and $g$ by Theorem~\ref{thm:rid_plus_parameter}, we allocate $10\%$ of the users randomly to estimate the distribution of $\pi_t$ via SW mechanism~\cite{li2020estimating}. Note that this is for parameter selection only, and will not inject any noise to the original data. %Then given a query over (combination) category size $d$, a specified privacy budget $\epsilon$ and the distribution of $\pi_t$, Theorem~\ref{thm:rid_plus_parameter} returns the optimal parameter setting of the numbers of dummies $m$, samples $s$ and groups $g$.

%For CTM-CRI, the given privacy will be first divided into two parts $\epsilon_1$ and $\epsilon_2$ equally for the estimation of $\bar{c}$ in Eq.~\ref{eq:baseline_fre_observed} and $\Delta\pi$ in Eq.~\ref{eq:pid_pi0} respectively. Note that, CTM-CRI may not ba able to satisfy a small privacy budget, especially when the category size $d$ is large (i.e., $\ln(d-1)>\frac{\epsilon}{2}$). In this case, we will adopt the grouping strategy of CTM-CRIAD to reduce the category size in CTM-CRI to achieve the given privacy budget. 
 
\noindent
{\bf Experiment Design.} We design three sets of experiments to evaluate different methods. The first set compares the overall performance of CRIAD and its competitors across three datasets by varying the privacy budget. The second set studies the impact of category sizes on the performance of different methods. The third set validates the effectiveness of Theorem~\ref{thm:rid_plus_parameter} for setting the optimal parameters for the numbers of dummies $m$, samples $s$ and groups $g$. 

%We design three sets of experiments. The first set compares the overall performance of four methods on three datasets through varying the privacy budget. The second set shows the performance of of four methods on three datasets through varying the size of query category. The third set explores the effectiveness of Theorem~\ref{thm:rid_plus_parameter} for choosing the optimal parameter setting of the numbers of dummies $m$, samples $s$ and groups $g$. 

\noindent
{\bf Metrics.}
To evaluate the result accuracy, we employ the Mean Relative Error ~\cite{mood1950introduction}, which quantifies the average difference between the estimated result $\tilde{Q}(c)$ and the ground truth $Q(c)$. Formally,

\begin{equation}
	MRE(c)=\frac{1}{N}\sum\nolimits_{N}\frac{\left|\tilde{Q}(c)-Q(c) \right| }{Q(c)}
\end{equation}
where $N$ represents the number of trials for each experiment. In our study, $N$ is set to $100$.

We conduct experiments using Python 3.11.5 and the Numpy 1.24.3 library on a desktop equipped with an Intel Core i5-13400F 1.50 GHz CPU and 64GB of RAM, running Windows 11.

\subsection{Overall Results}

In this subsection, we investigate the overall performance of different methods across three real-world datasets with varying privacy budgets. On each dataset, we evaluate three subset counting queries with categories set to $[1,100]$, $[1,400]$, and $[1,1600]$, respectively. The privacy budget varies from $0.2$ to $2.0$, with a step size of $0.2$. 
%It is worth noting that there is no result for PSP in Figure~\ref{fig:varying privacy budget}(b). This is because, when query category $c = [1,100]$, the Mean Relative Error of PSP is larger than $1$ for any given privacy budget on {\it OnlineRetail} dataset.
Figure~\ref{fig:varying privacy budget} shows the results, where MRE of all methods decreases as the privacy budget increases. Overall, CRIAD performs the best, followed by RR, NVP and finally PSP. 
The gap between CRIAD and the competitors is particularly significant for smaller privacy budget (e.g., $\epsilon<1.2$). 
This is because, a small privacy budget potentially leads to larger estimation variance, resulting in higher MRE. However, CRIAD is capable of selecting optimal parameters (e.g., by reducing the number of dummy items) to minimize the variance according to the given privacy budget, thanks to Theorem~\ref{thm:rid_plus_parameter}. 
In particular, with a small privacy budget of $\epsilon=0.1$, CRIAD outperforms the second-best method (i.e., RR) by a factor of at least 5.
On the other hand, we found PSP is much less effective compared to other methods. This is because, in each dataset, despite the large size of the category, the number of items owned by each user is small, which can lead to a large $\eta$ and ultimately result in large estimation error.

%when the privacy budget is small, the variance of CRIAD significantly impacts the final results. Under these circumstances, Theorem~\ref{thm:rid_plus_parameter} mitigates this by reducing $m$, thereby ensuring a smaller variance, despite this resulting in a certain degree of bias. Furthermore, PSP is much less effective relative to other methods. This is because, in each dataset, despite the large size of the category, the length of items owned by each user is small, which can lead to a large $\eta$ and ultimately cause an increase in variance.

\begin{figure*}[ht]
	\centering
	\begin{minipage}{0.55\linewidth}
		\centerline{\includegraphics[width=\linewidth]{figures/legend.pdf}}
		\vspace{0.05in}
	\end{minipage}
	
	\begin{minipage}{0.28\linewidth}
		\centerline{\includegraphics[width=\linewidth]{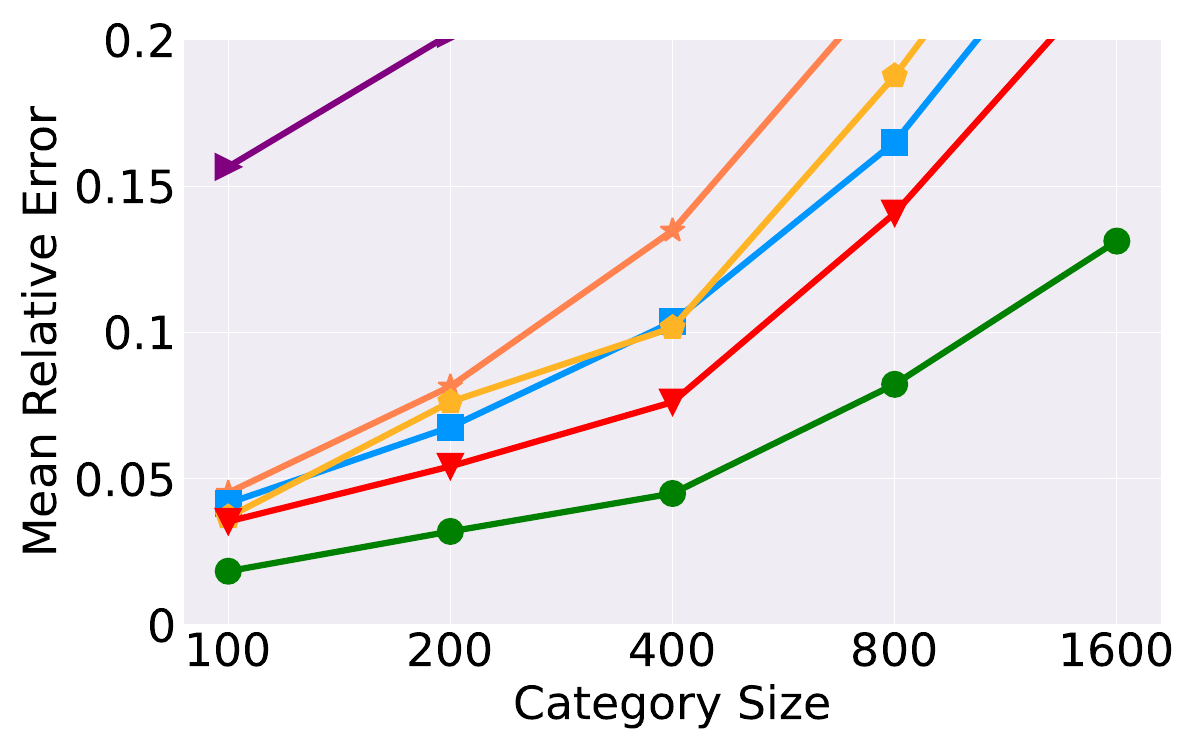}}
		\centerline{\footnotesize{(a) {\it Kosarak}}}
	\end{minipage}
	\begin{minipage}{0.28\linewidth}
		\centerline{\includegraphics[width=\linewidth]{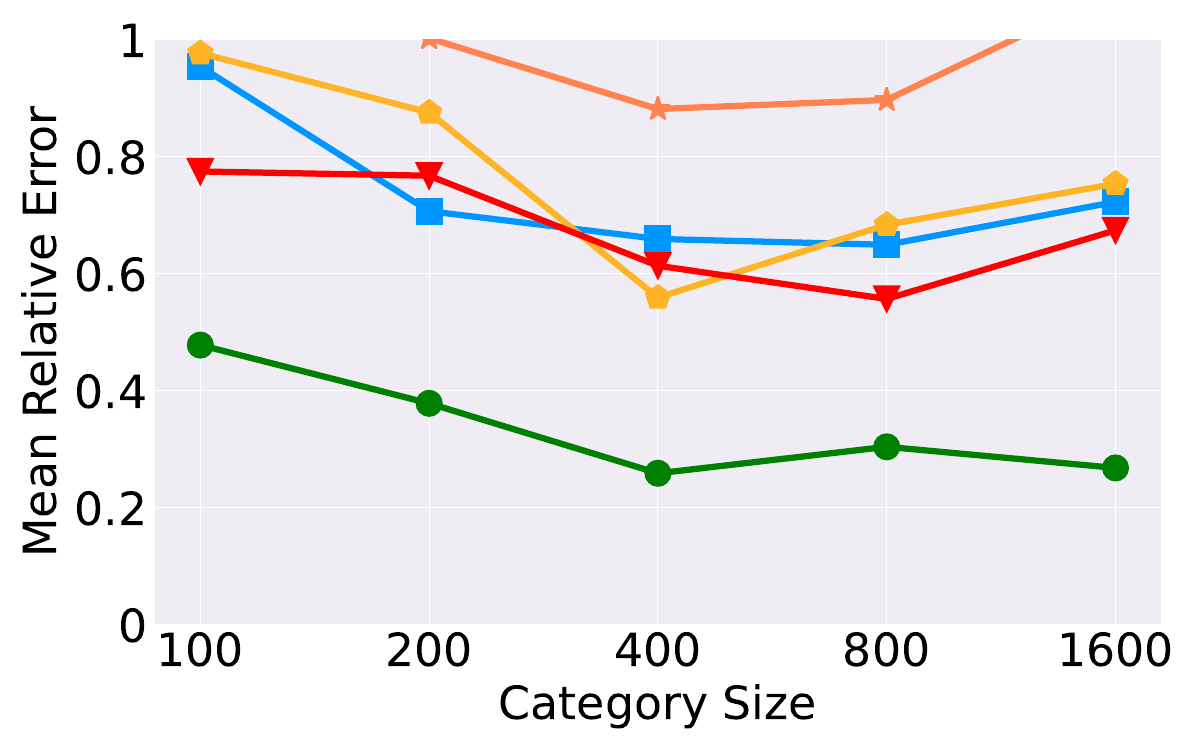}}
		\centerline{\footnotesize{(b) {\it OnlineRetail}}}
	\end{minipage}
	\begin{minipage}{0.28\linewidth}
		\centerline{\includegraphics[width=\linewidth]{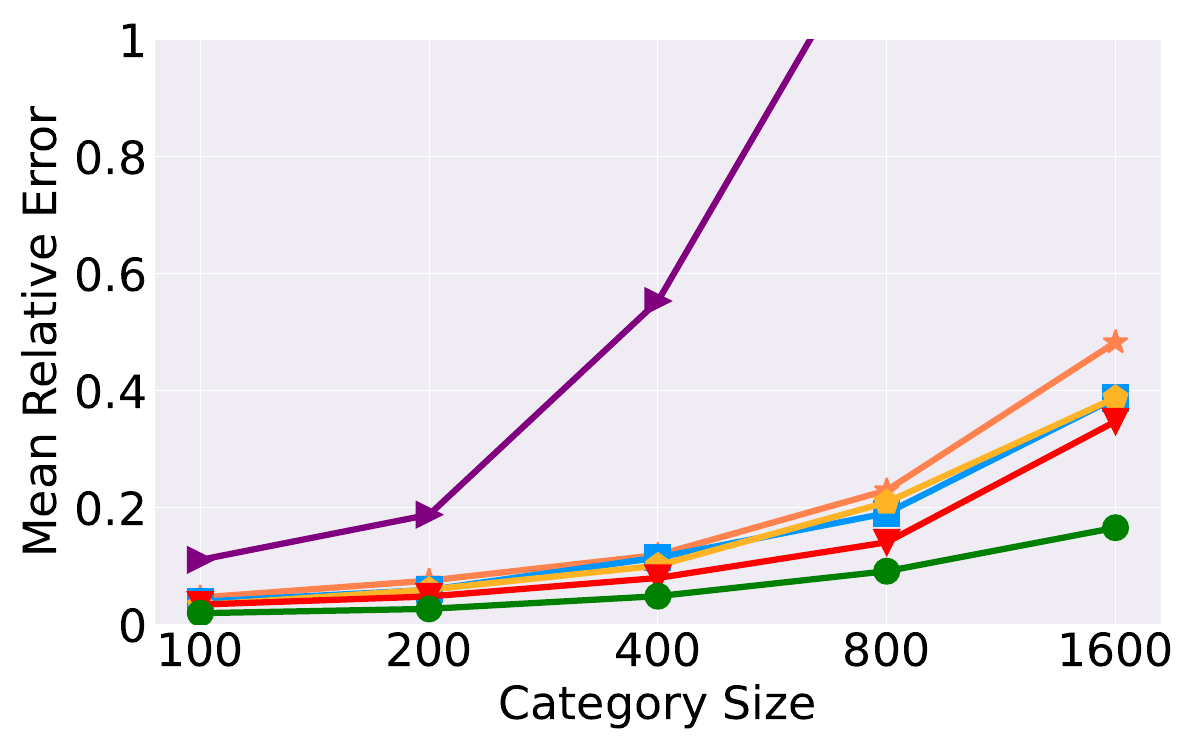}}
		\centerline{\footnotesize{(c) {\it POS}}}
	\end{minipage}
	
	\caption{MRE of different methods with varying category size ($\epsilon=1$).}
%	\vspace{-0.1in}
	\label{fig:varying size}
\end{figure*}

\subsection{Impact of Category Size}

In this subsection, we study the impact of the category size on the performance of different methods, while fixing the privacy budget at $\epsilon=1$. Specifically, the category size varies from 100 from 1600, resulting in the category domain from $[1,100]$ to $[1,1600]$ respectively. 
Figure~\ref{fig:varying size} presents the results across various category sizes, where CRIAD consistently delivers the best performance in all cases, while PSP always exhibits the highest MRE, which even exceeds $1$ in some cases, especially on {\it OnlineRetail} dataset. 

On the other hand, as shown in Figures~\ref{fig:varying size}(a) and \ref{fig:varying size}(c), we observe that the MRE of different methods increases with the category size on both {\it Kosarak} and {\it POS} datasets. This can be attributed to the distribution of items in both datasets, where the majority of items are concentrated in smaller category indexes. Consequently, as the category size increases, the number of items owned by users within the query category remains relatively stable, which ultimately leads to worse results. In {\it OnlineRetail} dataset, the items are more evenly distributed across users, resulting in a relatively consistent MRE across different category sizes. Overall, the performance gap between CRIAD and the other methods becomes more pronounced as the category size increases, which suggests that CRIAD scales effectively with larger category sizes.

\subsection{Effectiveness of the Optimal Parameter Setting}

In this subsection, we validate the effectiveness of Theorem~\ref{thm:rid_plus_parameter} for determining the optimal parameters, in terms of the numbers of dummies $m$, samples $s$ and groups $g$. For this study, we fix the privacy budget at $\epsilon=1$ and set the category size $d=400$. Note that the estimation of $\pi_t$ in Theorem~\ref{thm:rid_plus_parameter} introduces variability in parameter selection. Therefore, we conduct $100$ independent experiments across three datasets. Table~\ref{table:msgcount} presents the three most frequent $(m,s,g)$ combinations and their occurrences out of 100. We observe that, under the condition of $\epsilon=1$ and $d=400$, our strategy consistently sets $(m,s,g)=(148,1,1)$ with the highest probability. 

\begin{table}[h]
	\scriptsize
	\caption{Occurrences of the three most frequent parameter combinations out of 100 experiments}
	\label{table:msgcount}
	\centering
	\begin{tabular}{c|c|c|c}
		\hline
		$\bm{(m,s,g)}$            & {\bf Kosarak}     & {\bf OnlineRetail} & {\bf POS}         \\ \hline
		(148,1,1)          &60           & 63           & 58          \\ \hline
		(244,2,1)          & 9           & 4            & 3           \\ \hline
		(288,3,1)          & 2           & 2            & 5           \\ \hline
		%(355,8,1)          & 3           & 1            & 1           \\ \hline
	\end{tabular}
\end{table}

Then we fix $m$ at $148$, $244$, and $288$ respectively, and enumerate all feasible $(m,s,g)$ combinations that satisfy the given privacy budget. The results of MRE of each combination over three datasets are presented in Table~\ref{table:sg}, where the results of the selected parameter combinations by Theorem~\ref{thm:rid_plus_parameter} is shown as bold (i.e., parameter combinations in Table~\ref{table:msgcount}). Overall, we observe that those selected $(m,s,g)$ combinations yield a relatively lower relative error than most of the cases. 
Although the combination $(148,3,2)$ yields the smallest relative error, the difference is not substantial compared to the $(148,1,1)$ case, particularly on the {\it Kosarak} and {\it OnlineRetail} datasets. It is also noteworthy that when $g=2$, the computational cost almost doubles. Additionally, we observe that for fixed values of $m$ and $g$, a larger $s$ results in a decreasing MRE. This trend is explained by Theorem~\ref{thm:rid_plus_parameter} that, while satisfying $\epsilon$-LDP, a larger $s$ corresponds to a smaller expected squared error.

%Furthermore, given the frequent occurrence of certain $(m,s,g)$ combinations across all three datasets, we also fix the respective $m$ and iterated over all possible combinations. The results of these combinations are detailed in Table~\ref{table:sg} and~\ref{table:sg2}. It is readily apparent that Theorem~\ref{thm:rid_plus_parameter} consistently identifies the optimal parameter combinations when $m$ is some other value. Additionally, we observe that for fixed values of $m$ and $g$, a larger $s$ results in a decreasing MRE. This trend is explained by Theorem~\ref{thm:rid_plus_parameter}, which posits that under the same conditions satisfying $\epsilon$-LDP, a larger $s$ corresponds to a smaller expected squared error.

\begin{table}[h]
	\centering
	\scriptsize
	\caption{Results of MRE of parameter combinations over three dataset, with $d=400$, $\epsilon=1$.} 
	\label{table:sg}
	\begin{tabular}{c|c|ccc}
		\hline
		$\bm{m}$ & $\bm{(s,g)}$  & {\bf Kosarak}  & {\bf OnlineRetail}  & {\bf POS} \\
		\hline
		
		\multirow{4}{*}{$m=148$}
		& (1,1)&   \textbf{0.052} &\textbf{0.343}  &\textbf{0.060} \\
		\cline{2-5}
		& (1,2)  & 0.078          & 0.593          & 0.074 \\
		\cline{2-5}
		& (2,2) & 0.052          & 0.341          & 0.045 \\
		\cline{2-5}
		& (3,2) & 0.050  & 0.337     & 0.043 \\
		\hline
		
		\multirow{2}{*}{$m=244$}
		& (1,1)& 0.093 &0.596  &0.072 \\
		\cline{2-5}
		& (2,1)  & \textbf{0.369}  & \textbf{0.593}    & \textbf{0.070} \\
		\hline
		
		\multirow{3}{*}{$m=288$}
		& (1,1)& 0.109 &0.468  &0.082 \\
		\cline{2-5}
		& (2,1)& 0.067 &0.273  &0.074 \\
		\cline{2-5}
		& (3,1)  & \textbf{0.040}  & \textbf{0.199}    & \textbf{0.048} \\
		\hline
%		
%		\multirow{8}{*}{$m=355$}
%		& (1,1)& 0.069           &0.642          &0.086 \\
%		\cline{2-5}
%		& (2,1)  & 0.060          & 0.358          & 0.066 \\
%		\cline{2-5}
%		& (3,1) & 0.057          & 0.290         & 0.049 \\
%		\cline{2-5}
%		& (4,1) & 0.056          & 0.275          & 0.044 \\
%		\cline{2-5}
%		& (5,1) & 0.037        & 0.240          & 0.031 \\
%		\cline{2-5}
%		& (6,1) & 0.036          & 0.231          & 0.030 \\
%		\cline{2-5}
%		& (7,1) & 0.034          & 0.229       & 0.027 \\
%		\cline{2-5}
%		& (8,1) & \textbf{0.033}  & \textbf{0.228}     & \textbf{0.025} \\
%		\hline
	\end{tabular}
\end{table}

To further investigate the effectiveness of parameter setting, we then fix $s=1$ or $g=1$ respectively, and randomly select $100$ combinations of $(m,s,g)$ that satisfy the given privacy budget $\epsilon=1$. The results are presented in Figure~\ref{fig:msg}, where combinations with larger MRE than the case of $(148,1,1)$ are marked in red, while those with smaller MRE are in blue.

%\begin{figure*}[h]
%	\centering
%	
%	\begin{minipage}{0.25\linewidth}
%		\centerline{\includegraphics[width=\linewidth]{figures/kosarak/kosarak_mg.pdf}}
%		\centerline{\footnotesize{(a) {\it Kosarak}}}
%	\end{minipage}
%	\begin{minipage}{0.25\linewidth}
%		\centerline{\includegraphics[width=\linewidth]{figures/OnlineRetail/OnlineRetail_mg.pdf}}
%		\centerline{\footnotesize{(b) {\it OnlineRetail}}}
%	\end{minipage}
%	\begin{minipage}{0.25\linewidth}
%		\centerline{\includegraphics[width=\linewidth]{figures/pos/pos_mg.pdf}}
%		\centerline{\footnotesize{(c) {\it POS}}}
%	\end{minipage}
%	
%	\caption{MRE of 100 random parameter combinations of $(m,s,g)$, with $\epsilon=1$ and $s=1$.}
%	\label{fig:mg}
%\end{figure*}
%
%
%\begin{figure*}[h]
%	\centering
%
%	\begin{minipage}{0.25\linewidth}
%		\centerline{\includegraphics[width=\linewidth]{figures/kosarak/kosarak_ms.pdf}}
%		\centerline{\footnotesize{(a) {\it Kosarak} }}
%	\end{minipage}
%	\begin{minipage}{0.25\linewidth}
%		\centerline{\includegraphics[width=\linewidth]{figures/OnlineRetail/OnlineRetail_ms.pdf}}
%		\centerline{\footnotesize{(b) {\it OnlineRetail}}}
%	\end{minipage}
%	\begin{minipage}{0.25\linewidth}
%		\centerline{\includegraphics[width=\linewidth]{figures/pos/pos_ms.pdf}}
%		\centerline{\footnotesize{(c) {\it POS}}}
%	\end{minipage}
%	
%	\caption{MRE of 100 random parameter combinations of $(m,s,g)$, with $\epsilon=1$ and $g=1$.}
%	\vspace{-0.1in}
%	\label{fig:ms}
%\end{figure*}

\begin{figure}[h]
	\centering	
	\begin{minipage}{0.49\linewidth}
		\centerline{\includegraphics[width=\linewidth]{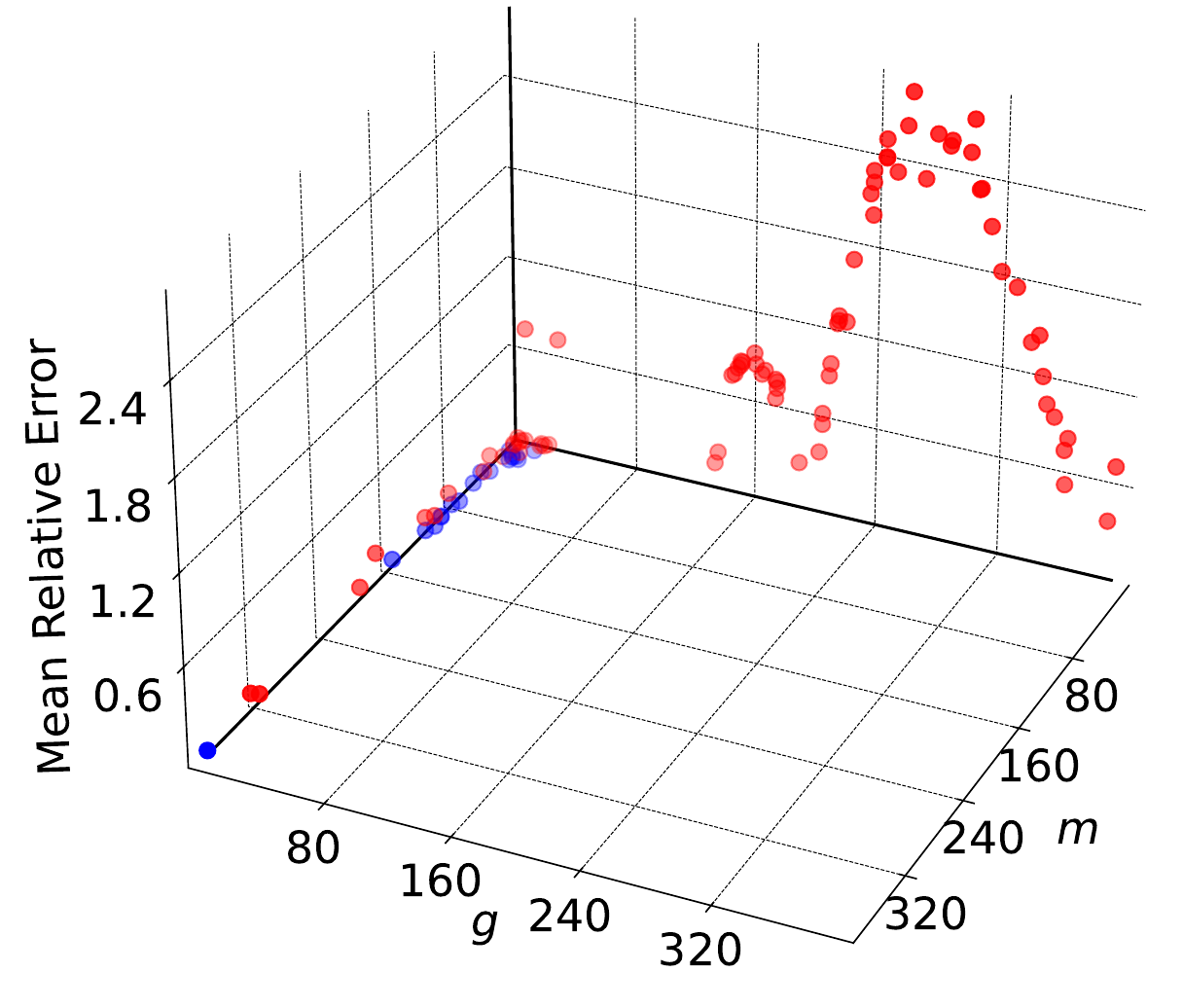}}
		\centerline{\footnotesize{(a) {\it Kosarak}, $\epsilon=1$, $s=1$}}
	\end{minipage}
	\begin{minipage}{0.49\linewidth}
		\centerline{\includegraphics[width=\linewidth]{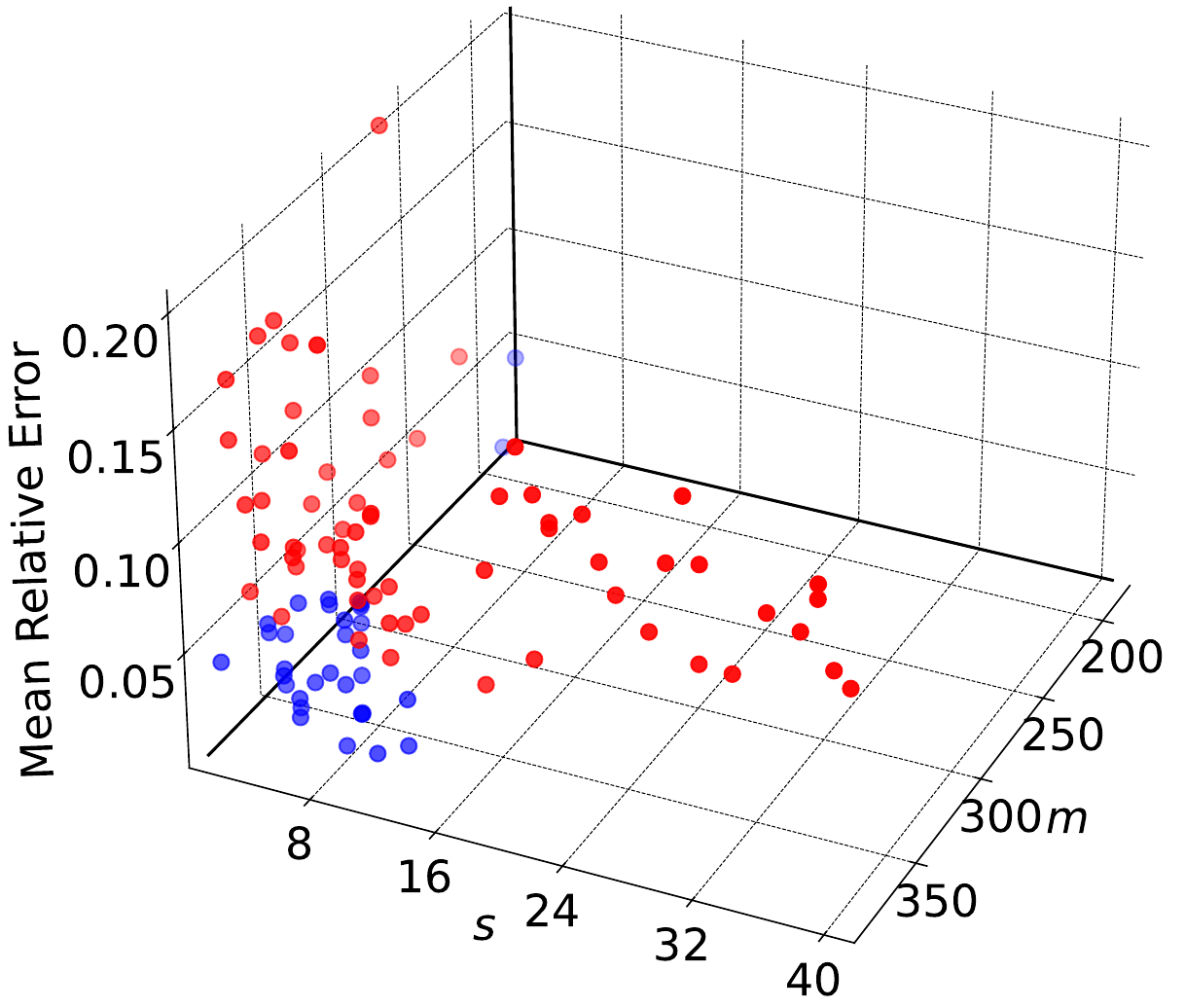}}
		\centerline{\footnotesize{(b) {\it Kosarak}, $\epsilon=1$, $g=1$ }}
	\end{minipage}

	\begin{minipage}{0.49\linewidth}
		\centerline{\includegraphics[width=\linewidth]{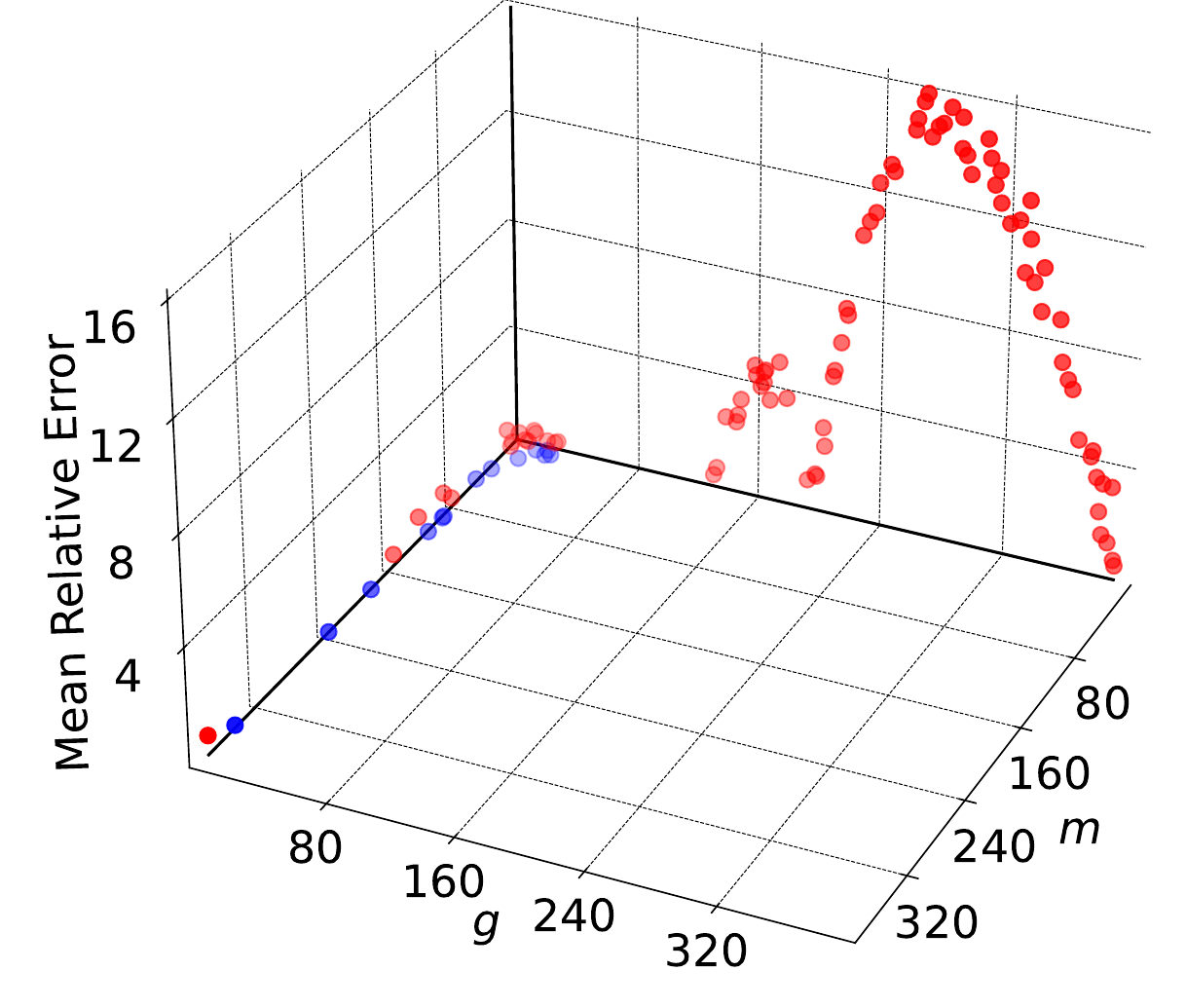}}
		\centerline{\footnotesize{(c) {\it OnlineRetail}, $\epsilon=1$, $s=1$}}
	\end{minipage}
	\begin{minipage}{0.49\linewidth}
		\centerline{\includegraphics[width=\linewidth]{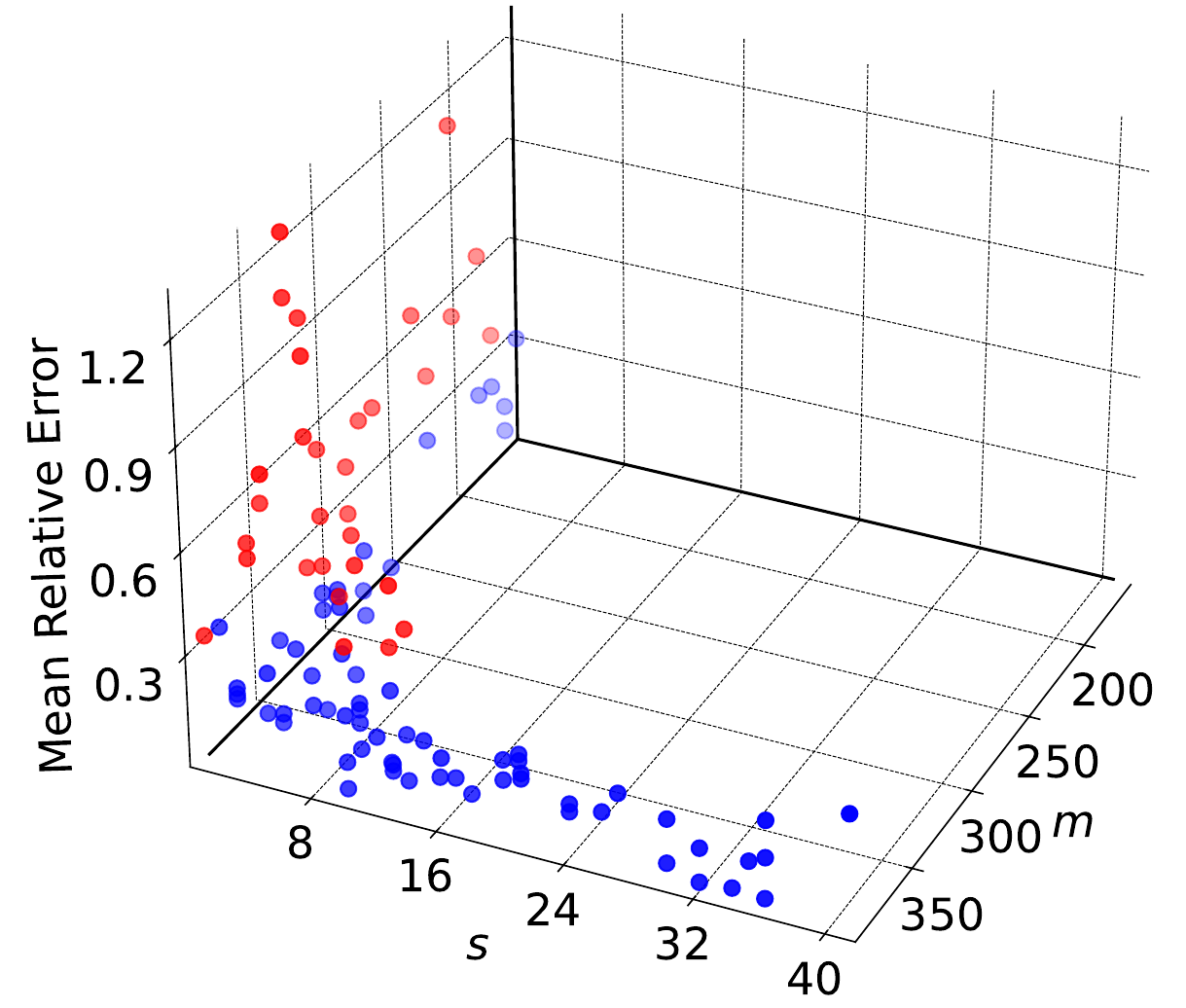}}
		\centerline{\footnotesize{(d) {\it OnlineRetail}, $\epsilon=1$, $g=1$}}
	\end{minipage}

	\begin{minipage}{0.49\linewidth}
		\centerline{\includegraphics[width=\linewidth]{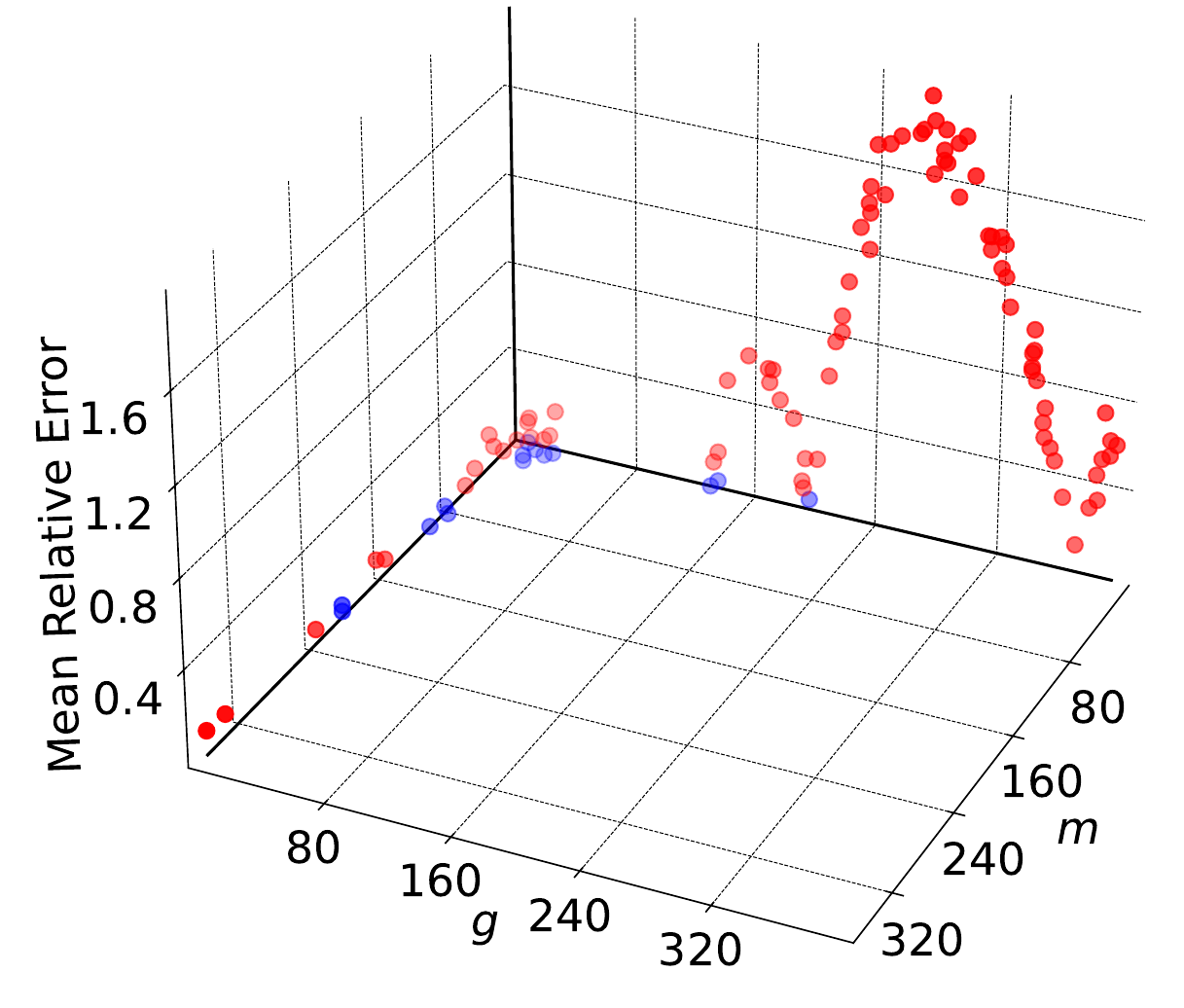}}
		\centerline{\footnotesize{(e) {\it POS}, $\epsilon=1$, $s=1$}}
	\end{minipage}
	\begin{minipage}{0.49\linewidth}
		\centerline{\includegraphics[width=\linewidth]{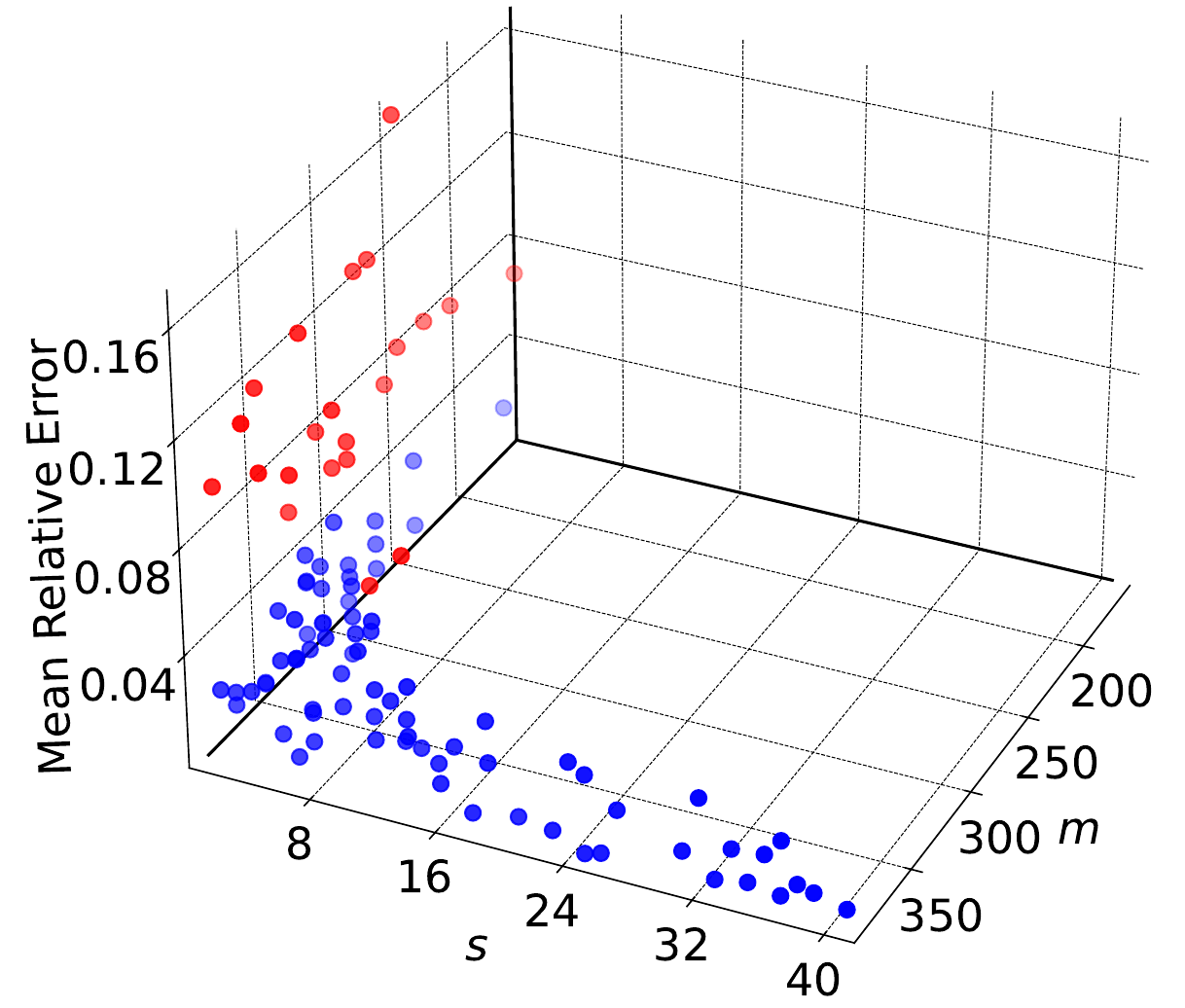}}
		\centerline{\footnotesize{(f) {\it POS}, $\epsilon=1$, $g=1$}}
	\end{minipage}	
	\caption{MRE of 100 random parameter combinations of $(m,s,g)$.}
	\label{fig:msg}
\end{figure}

As shown in Figures~\ref{fig:msg}(a), (c), and (e), we find that when fixing the number of samples at $s=1$, setting a smaller number of groups $g$ tends to achieve a lower MRE. This is to ensure more valid samples for the estimation within each group. As $g$ increases, the number of dummy items $m$ is significantly reduced, since it is constrained by the decreasing group size. Nevertheless, the case of $(148,1,1)$ selected by Theorem~\ref{thm:rid_plus_parameter} still achieves lower MRE than most of parameter combinations.  

Figure~\ref{fig:msg}(b) shows that when $m>300$ and $s<16$, the corresponding parameter combinations outperform $(148,1,1)$ on {\it Kosarak} dataset. Additionally, as evidenced by Table~\ref{table:msgcount}, besides the $(148,1,1)$ combination, Theorem~\ref{thm:rid_plus_parameter} also frequently sets parameters within this range with a high probability on {\it Kosarak}. Figures~\ref{fig:msg}(d) and (f) demonstrate that on {\it  OnlineRetail} and {\it POS} datasets, parameter combinations with $m>350$ and $s<40$ also achieve lower MRE than $(148,1,1)$. Similarly, Table~\ref{table:msgcount} indicates that, besides the combination$(148,1,1)$, Theorem~\ref{thm:rid_plus_parameter} frequently sets parameters within this specified range with considerable probability. Overall, although Theorem~\ref{thm:rid_plus_parameter} does not always determine the optimal setting when $g$ is fixed, it reliably suggests better-than-average parameter configurations and frequently identifies optimal combinations.

\section{Related Work}
\label{sec:related_work}
%In this section, we review relevant existing works on local differential privacy (LDP), namely LDP mechanisms for categorical values, numerical values and set-value data, respectively.  

Differential privacy was first proposed in the centralized setting~\cite{dwork2006calibrating,dwork2014algorithmic,fu2024dpsur}.
To avoid relying on a trusted data collector, local differential privacy (LDP) was proposed to let each user perturb her data locally~\cite{duchi2013local}. In the literature, many LDP techniques have been proposed for various statistical collection tasks, such as frequency of categorical values~\cite{erlingsson2014rappor, kairouz2014extremal, bassily2015local, wang2017locally}, and mean of numerical values~\cite{ding2017collecting, wang2019collecting}. Recently, the research focus in LDP has been shifted to more complex tasks, such as heavy hitter identification~\cite{bassily2017practical, bun2018heavy}, itemset mining~\cite{wang2018locally}, marginal release~\cite{cormode2018practical, zhang2018calm}, time series data analysis~\cite{ye2023stateful,wang2021continuous,bao2021cgm}, and data poisoning attacks~\cite{huang2024ldpguard, du2023differential}. In what follows, we review existing LDP works that are relevant to ours, namely LDP mechanisms for categorical data, numerical data and set-value data, respectively.   

{\bf LDP Mechanisms for Categorical Data}. 
There are a line of LDP work developed for categorical value perturbation. Randomized Response (RR)~\cite{warner1965randomized} is the most straightforward mechanism for binary value, and a generalized form of RR (a.k.a., kRR)~\cite{kairouz2014extremal} is then proposed to deal with a value with domain size $d>2$. To alleviate large perturbation noise along with the increasing domain size, Wang \emph{et al.} propose Optimized Unary Encoding (OUE)~\cite{wang2017locally} which achieves better utility. Besides, several perturbation protocols are also proposed in the literature, including RAPPOR~\cite{erlingsson2014rappor}, SHist~\cite{bassily2015local} and subset selection~\cite{ye2018optimal}. 

{\bf LDP Mechanisms for Numerical Data}. 
Designing LDP protocols for numerical values also attracts great attention from the researchers. As with centralized DP, Laplace Mechanism~\cite{dwork2006calibrating} can be adopted in the local setting. On the other hand, Duchi \emph{et al.} propose a solution for mean estimation over numerical values~\cite{duchi2014privacy}. To address computation and space complexity of it, an improved method~\cite{duchi2018minimax} is then proposed to perturb any numerical input into a binary output according to a certain probability. Ding \emph{et al.}~\cite{ding2017collecting} also propose mechanisms to continuously collect telemetry data. More recently, Wang \emph{et al.}~\cite{wang2019collecting} propose Piecewise Mechanism (PM) to improve estimation accuracy, and Li \emph{et al.}~\cite{li2020estimating} propose Square-wave (SW) mechanism to support numerical distribution estimation.  

{\bf Frequency Estimation over Set-value Data}. As a relevant problem to this paper, frequency estimation over set-value data has been widely studies in the context of LDP. 
LDPMiner~\cite{qin2016heavy}, together with a \emph{padding-and-sampling} protocol, is the first solution for this problem. Wang \emph{et al.} \cite{wang2018locally} further study the privacy amplification effect of padding-and-sampling protocol with kRR, and use it for frequent itemset mining. Wang \emph{et al.}~\cite{wang2018privset} propose PrivSet to estimate both item distribution and set size distribution over set-value data. 
Besides, some other work also consider the setting where each user possesses a set of key-value pairs~\cite{ye2019privkv, gu2020pckv}. 
Nevertheless, existing work over set-value setting all consider a full-domain item distribution or heavy hitter identification~\cite{wang2019locally, du2024top}, both of which are different from the problem studied in this work.

\section{Conclusion}
\label{sec:conclusion}

In this work, we study the problem of answering subset counting queries on set-value data. Unlike existing studies that perturb the original values to ensure an LDP guarantee, we propose an alternative approach that leverages the deniability of randomized indexes without perturbing the values. Our design, named CRIAD, satisfies a rigorous LDP guarantee while achieving higher accuracy than all existing methods. Furthermore, by integrating a multi-dummy, multi-sample, and multi-group strategy, CRIAD is optimized into a fully scalable solution that offers flexibility across various privacy requirements and domain sizes. 

As for further work, we plan to extend CRIAD to more complex scenarios, such as the federated setting, where queries involve users across multiple parties, and item exhibit heterogeneity.




%\bibliographystyle{IEEEtran}

\bibliographystyle{abbrv}
\bibliography{ref}

\begin{thebibliography}{10}

\bibitem{bao2021cgm}
E.~Bao, Y.~Yang, X.~Xiao, and B.~Ding.
\newblock {CGM}: an enhanced mechanism for streaming data collection with local
  differential privacy.
\newblock {\em PVLDB}, 14(11):2258--2270, 2021.

\bibitem{bassily2015local}
R.~Bassily and A.~Smith.
\newblock Local, private, efficient protocols for succinct histograms.
\newblock In {\em STOC}, pages 127--135. ACM, 2015.

\bibitem{bassily2017practical}
R.~Bassily, U.~Stemmer, A.~G. Thakurta, et~al.
\newblock Practical locally private heavy hitters.
\newblock In {\em NIPS}, pages 2285--2293, 2017.

\bibitem{bun2018heavy}
M.~Bun, J.~Nelson, and U.~Stemmer.
\newblock Heavy hitters and the structure of local privacy.
\newblock In {\em PODS}, pages 435--447. ACM, 2018.

\bibitem{cormode2018practical}
G.~Cormode, T.~Kulkarni, and D.~Srivastava.
\newblock Marginal release under local differential privacy.
\newblock In {\em SIGMOD}, pages 131--146. ACM, 2018.

\bibitem{ding2017collecting}
B.~Ding, J.~Kulkarni, and S.~Yekhanin.
\newblock Collecting telemetry data privately.
\newblock In {\em NIPS}, pages 3574--3583, 2017.

\bibitem{du2024top}
R.~Du, Q.~Ye, Y.~Fu, H.~Hu, and K.~Huang.
\newblock Top-k discovery under local differential privacy: An adaptive
  sampling approach.
\newblock {\em IEEE Transactions on Dependable and Secure Computing}, 2024.

\bibitem{du2023differential}
R.~Du, Q.~Ye, Y.~Fu, H.~Hu, J.~Li, C.~Fang, and J.~Shi.
\newblock Differential aggregation against general colluding attackers.
\newblock In {\em ICDE}, pages 2180--2193. IEEE, 2023.

\bibitem{duchi2013local}
J.~C. Duchi, M.~I. Jordan, and M.~J. Wainwright.
\newblock Local privacy and statistical minimax rates.
\newblock In {\em FOCS}, pages 429--438. IEEE, 2013.

\bibitem{duchi2014privacy}
J.~C. Duchi, M.~I. Jordan, and M.~J. Wainwright.
\newblock Privacy aware learning.
\newblock {\em Journal of the ACM}, 61(6):1--57, 2014.

\bibitem{duchi2018minimax}
J.~C. Duchi, M.~I. Jordan, and M.~J. Wainwright.
\newblock Minimax optimal procedures for locally private estimation.
\newblock {\em Journal of the American Statistical Association},
  113(521):182--201, 2018.

\bibitem{dwork2006differential}
C.~Dwork.
\newblock Differential privacy.
\newblock In {\em ICALP}, pages 1--12. Springer, 2006.

\bibitem{dwork2006calibrating}
C.~Dwork, F.~McSherry, K.~Nissim, and A.~Smith.
\newblock Calibrating noise to sensitivity in private data analysis.
\newblock In {\em TCC}, pages 265--284. Springer, 2006.

\bibitem{dwork2014algorithmic}
C.~Dwork, A.~Roth, et~al.
\newblock The algorithmic foundations of differential privacy.
\newblock {\em Foundations and Trends{\textregistered} in Theoretical Computer
  Science}, 9(3--4):211--407, 2014.

\bibitem{erlingsson2014rappor}
{\'U}.~Erlingsson, V.~Pihur, and A.~Korolova.
\newblock Rappor: Randomized aggregatable privacy-preserving ordinal response.
\newblock In {\em CCS}, pages 1054--1067. ACM, 2014.

\bibitem{fu2024dpsur}
J.~Fu, Q.~Ye, H.~Hu, Z.~Chen, L.~Wang, K.~Wang, and X.~Ran.
\newblock Dpsur: accelerating differentially private stochastic gradient
  descent using selective update and release.
\newblock {\em Proceedings of the VLDB Endowment}, 17(6):1200--1213, 2024.

\bibitem{gu2020pckv}
X.~Gu, M.~Li, Y.~Cheng, L.~Xiong, and Y.~Cao.
\newblock {PCKV:} locally differentially private correlated key-value data
  collection with optimized utility.
\newblock In {\em USENIX Security}, 2020.

\bibitem{huang2024ldpguard}
K.~Huang, G.~Ouyang, Q.~Ye, H.~Hu, B.~Zheng, X.~Zhao, R.~Zhang, and X.~Zhou.
\newblock {LDPGuard}: Defenses against data poisoning attacks to local
  differential privacy protocols.
\newblock {\em IEEE Transactions on Knowledge and Data Engineering},
  36(7):3195--3209, 2024.

\bibitem{kairouz2014extremal}
P.~Kairouz, S.~Oh, and P.~Viswanath.
\newblock Extremal mechanisms for local differential privacy.
\newblock In {\em NIPS}, pages 2879--2887, 2014.

\bibitem{kasiviswanathan2011can}
S.~P. Kasiviswanathan, H.~K. Lee, K.~Nissim, S.~Raskhodnikova, and A.~Smith.
\newblock What can we learn privately?
\newblock {\em SIAM Journal on Computing}, 40(3):793--826, 2011.

\bibitem{li2016differential}
N.~Li, M.~Lyu, D.~Su, and W.~Yang.
\newblock Differential privacy: From theory to practice.
\newblock {\em Synthesis Lectures on Information Security, Privacy, \& Trust},
  8(4):1--138, 2016.

\bibitem{li2020estimating}
Z.~Li, T.~Wang, M.~Lopuha{\"a}-Zwakenberg, N.~Li, and B.~{\v{S}}koric.
\newblock Estimating numerical distributions under local differential privacy.
\newblock In {\em SIGMOD}, pages 621--635, 2020.

\bibitem{mcsherry2009privacy}
F.~McSherry.
\newblock Privacy integrated queries: an extensible platform for
  privacy-preserving data analysis.
\newblock In {\em SIGMOD}, pages 19--30. ACM, 2009.

\bibitem{mood1950introduction}
A.~M. Mood.
\newblock Introduction to the theory of statistics.
\newblock 1950.

\bibitem{qin2016heavy}
Z.~Qin, Y.~Yang, T.~Yu, I.~Khalil, X.~Xiao, and K.~Ren.
\newblock Heavy hitter estimation over set-valued data with local differential
  privacy.
\newblock In {\em CCS}, pages 192--203. ACM, 2016.

\bibitem{thakurta2017learning}
A.~G. Thakurta, A.~H. Vyrros, U.~S. Vaishampayan, G.~Kapoor, J.~Freudiger,
  V.~R. Sridhar, and D.~Davidson.
\newblock Learning new words, Mar.~14 2017.
\newblock US Patent 9,594,741.

\bibitem{wang2019collecting}
N.~Wang, X.~Xiao, Y.~Yang, J.~Zhao, S.~C. Hui, H.~Shin, J.~Shin, and G.~Yu.
\newblock Collecting and analyzing multidimensional data with local
  differential privacy.
\newblock In {\em ICDE}, 2019.

\bibitem{wang2018privset}
S.~Wang, L.~Huang, Y.~Nie, P.~Wang, H.~Xu, and W.~Yang.
\newblock {PrivSet}: Set-valued data analyses with locale differential privacy.
\newblock In {\em INFOCOM}, pages 1088--1096. IEEE, 2018.

\bibitem{wang2017locally}
T.~Wang, J.~Blocki, N.~Li, and S.~Jha.
\newblock Locally differentially private protocols for frequency estimation.
\newblock In {\em USENIX Security}, pages 729--745, 2017.

\bibitem{wang2021continuous}
T.~Wang, J.~Q. Chen, Z.~Zhang, D.~Su, Y.~Cheng, Z.~Li, N.~Li, and S.~Jha.
\newblock Continuous release of data streams under both centralized and local
  differential privacy.
\newblock In {\em CCS}, pages 1237--1253, 2021.

\bibitem{wang2018locally}
T.~Wang, N.~Li, and S.~Jha.
\newblock Locally differentially private frequent itemset mining.
\newblock In {\em S\&P}, pages 127--143. IEEE, 2018.

\bibitem{wang2019locally}
T.~Wang, N.~Li, and S.~Jha.
\newblock Locally differentially private heavy hitter identification.
\newblock {\em IEEE Transactions on Dependable and Secure Computing},
  18(2):982--993, 2019.

\bibitem{warner1965randomized}
S.~L. Warner.
\newblock Randomized response: A survey technique for eliminating evasive
  answer bias.
\newblock {\em Journal of the American Statistical Association},
  60(309):63--69, 1965.

\bibitem{ye2018optimal}
M.~Ye and A.~Barg.
\newblock Optimal schemes for discrete distribution estimation under locally
  differential privacy.
\newblock {\em IEEE Transactions on Information Theory}, 64(8):5662--5676,
  2018.

\bibitem{ye2023stateful}
Q.~Ye, H.~Hu, K.~Huang, M.~H. Au, and Q.~Xue.
\newblock Stateful switch: Optimized time series release with local
  differential privacy.
\newblock In {\em INFOCOM}. IEEE, 2023.

\bibitem{ye2019privkv}
Q.~Ye, H.~Hu, X.~Meng, and H.~Zheng.
\newblock Priv{KV}: Key-value data collection with local differential privacy.
\newblock In {\em S\&P}, pages 317--331. IEEE, 2019.

\bibitem{zhang2018calm}
Z.~Zhang, T.~Wang, N.~Li, S.~He, and J.~Chen.
\newblock {CALM}: Consistent adaptive local marginal for marginal release under
  local differential privacy.
\newblock In {\em CCS}, pages 212--229. ACM, 2018.

\bibitem{zheng2001real}
Z.~Zheng, R.~Kohavi, and L.~Mason.
\newblock Real world performance of association rule algorithms.
\newblock In {\em SIGKDD}, pages 401--406, 2001.

\end{thebibliography}

%\input{meta-review}

\end{document}